\documentclass[a4size, 10pt,final,printed]{IEEEtran}
\usepackage{amsmath}
\usepackage{amssymb}
\usepackage{amsfonts}
\usepackage{graphicx}
\usepackage{epsfig}
\usepackage{subfigure}
\usepackage{psfrag}
\usepackage{algorithm, algorithmic}

\usepackage{booktabs}
\usepackage{ctable}
\usepackage{color}

 \title{Spatial Throughput Maximization of Wireless Powered Communication Networks}
 \author{Yue Ling Che,  \emph{Member, IEEE},  Lingjie Duan, \emph{Member, IEEE},  and Rui Zhang, \emph{Member, IEEE}
 \thanks{This work was supported by the SUTD-MIT International Design Centre (IDC) Grant (Project Number: IDSF1200106OH) and the SUTD-ZJU Joint Collaboration Grant (Project Number: ZJUR 041407). This work was also supported in part by the National University of Singapore under Research Grant R-263-000-679-133. Lingjie Duan is the corresponding author.}
 \thanks{Y.~L.~Che and L.~Duan are with the Engineering Systems and Design
 Pillar, Singapore University of Technology and Design (e-mail:
 yueling\_che@sutd.edu.sg;  lingjie\_duan@sutd.edu.sg).}
 \thanks{R.~Zhang is with the Department of Electrical and Computer Engineering,
 National University of Singapore (e-mail: elezhang@nus.edu.sg). He is also with
 the Institute for Infocomm Research, A*STAR, Singapore.}
 }

\begin{document}
\maketitle
\thispagestyle{empty}

\begin{abstract}

Wireless charging is a promising way to power wireless nodes'   transmissions.
This paper considers new dual-function access points (APs)
which are able to support the energy/information transmission to/from wireless
nodes. We focus on a large-scale wireless powered communication network (WPCN),
and use stochastic geometry to analyze the wireless nodes' performance
tradeoff between energy harvesting and information transmission.
 We   study two cases with
battery-free and battery-deployed wireless nodes.  For both cases, we   consider
a harvest-then-transmit protocol by partitioning each time frame into a downlink
(DL) phase for energy transfer, and an uplink (UL) phase for information
transfer.   By jointly optimizing frame partition between the two phases and the
wireless nodes' transmit power, we maximize the wireless nodes' spatial
throughput subject to a successful information transmission probability constraint.
For the battery-free case, we show that the wireless nodes prefer to choose small
transmit power  to obtain large transmission  opportunity. For the battery-deployed case, we first study an
ideal infinite-capacity battery scenario for wireless nodes, and show that the optimal charging design is not unique,
due to the sufficient energy stored in the battery. We then
extend to the practical finite-capacity battery scenario.  Although the exact
performance is difficult to be obtained analytically, it is  shown to be
upper and lower bounded by those  in the
infinite-capacity battery scenario and the battery-free case, respectively.
Finally, we provide numerical
results to corroborate our study.

\end{abstract}


\newtheorem{definition}{\underline{Definition}}[section]
\newtheorem{fact}{Fact}
\newtheorem{assumption}{Assumption}
\newtheorem{theorem}{\underline{Theorem}}[section]
\newtheorem{lemma}{\underline{Lemma}}[section]
\newtheorem{corollary}{\underline{Corollary}}[section]
\newtheorem{proposition}{\underline{Proposition}}[section]
\newtheorem{example}{\underline{Example}}[section]
\newtheorem{remark}{\underline{Remark}}[section]
\newcommand{\mv}[1]{\mbox{\boldmath{$ #1 $}}}
\newtheorem{property}{\underline{Property}}[section]

\section{Introduction}
By enabling the wireless devices to scavenge energy from the environment, energy harvesting has become a
promising solution to provide perpetual lifetime for energy-constrained wireless networks (e.g., the wireless sensor networks)  \cite{microwave}.
In particular, with the ability to cater to the mobility of the wireless nodes,
the ambient radio-frequency (RF) signals have been considered as a vital and
widely available  energy resource to power  wireless communication networks  \cite{Visser.IEEEProc.2013}.
 In recent point-to-point energy transfer experiments \cite{EH_overview}, wireless power of  3.5mW and 1uW have been harvested from the
 RF signals at distances of 0.6 and 11 meters, respectively.  Moreover, in the experiment-based study in \cite{He.TMC.13},
 the harvested energy from multiple energy transmitting sources is shown to be additive, which can be exploited to extend the operation range of wireless charging.
Due to  the appealing features of the RF-based energy harvesting,  the
{\it wireless powered communication network} (WPCN) \cite{Ju.TWC},
in which the wireless nodes exploit the harvested RF energy to power their
information transmissions, has  attracted
growing  attentions.

Different from traditional wireless networks, where the wireless nodes can draw
energy from reliable power supplies (e.g., by connecting to the power grid or a
battery),
due to the wireless fading channels, the random movement  of the wireless nodes,
as well as the employed energy harvesting techniques, the amount of
energy that can be harvested in a WPCN   is
generally uncertain. As a result, to meet   the quality-of-service (QoS)
requirement of the information transmission, the designed transmission schemes must be  adaptive  to the
dynamics of the harvested RF energy. Although challenging,  by assuming  completely or partially known knowledge of the
energy arrival processes, effective transmission schemes have been proposed in, e.g.,  \cite{Ozel.JSAC.2011}-\cite{Xu.JSAC.2014}.
 However,   the adopted energy arrival models in the above studies do not apply to the RF-based energy
harvesting scenario.

\begin{figure*}
\setlength{\abovecaptionskip}{-0.04in}
\centering
\DeclareGraphicsExtensions{.eps,.mps,.pdf,.jpg,.png}
\DeclareGraphicsRule{*}{eps}{*}{}
\includegraphics[angle=0, width=0.8\textwidth]{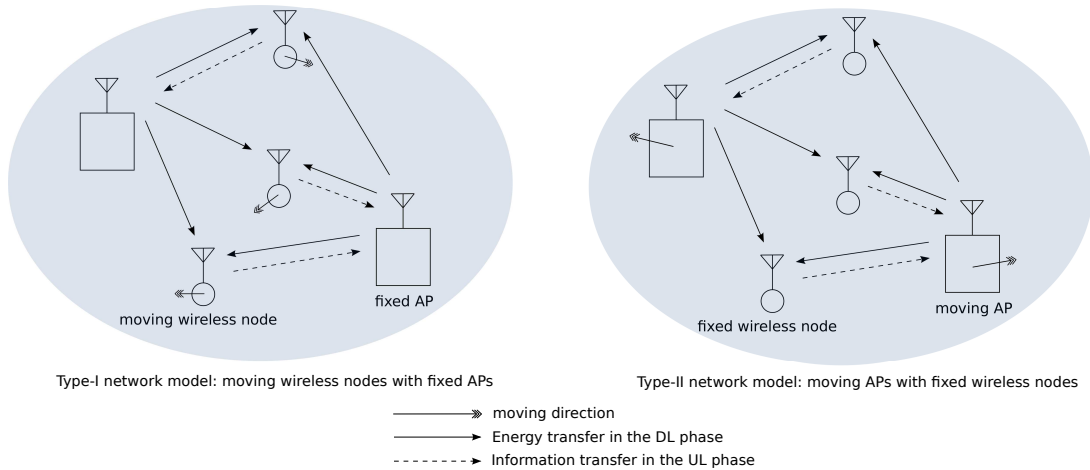}
\caption{Two types of WPCN models with DL energy
harvesting  and UL information transmission.}
\label{fig: network_model}
\vspace{-0.18in}
\end{figure*}

There has been a growing research interest focusing on  a point-to-point
or point-to-multipoint RF energy harvesting system, where a single transmitter transmits energy to a
single wireless node or multiple wireless nodes, respectively (e.g., in  \cite{Ju.TWC}, \cite{Rui.TWC.13}, and \cite{Liu.TWC.13}).
In particular,  in  \cite{Ju.TWC} the authors studied a
point-to-multipoint system, where  the energy transfer from an access
point (AP) to multiple wireless nodes is separated from  the information transfer
from each of the wireless nodes to the AP  in time domain.   By
exploiting the harvested energy at each wireless node, \cite{Ju.TWC}
investigated the optimal time allocation for  energy transfer and
information transfer, so as to maximize the system throughput with fairness consideration.
Moreover, since the RF signals may also carry information besides energy, simultaneous
wireless information and power transfer (SWIPT)
has been  studied in the literature  (see e.g.  \cite{Rui.TWC.13},
\cite{Liu.TWC.13}), where more sophisticated receiver design is involved.
In addition, we also noticed there are some works focusing on energy-efficient design for other applications
 (e.g., \cite{Wen.datacenter1}-\cite{Wen.EE}).

However,  most of the existing work, including the above
mentioned ones, did not consider optimal transmission scheme design in a
\emph{large-scale} WPCN with a very large number of wireless nodes, mainly due to the
following reasons: 1) it is difficult to design a  scalable
transmission scheme that can be efficiently implemented with the increasing number of wireless nodes;
and 2) due to the wireless fading channels as well as
the random placement of both energy transmitters and  wireless nodes,
it is challenging to analytically characterize the harvested RF energy by
a wireless node from multiple energy transmitters.
It came to our attention that stochastic geometry, as a novel way to analyze
 large-scale  communication networks, provides a set of powerful
mathematical tools for modeling and designing the wireless networks \cite{Stoyan.SG.95}.
Moreover,  the mathematical tools (e.g., probability generating functional
(PGFL) of a Poisson point process (PPP)), which facilitate the interference analysis  in a
wireless communication network \cite{Haenggi.book}, can also help
characterize the harvested RF energy in a WPCN~\cite{Huang.TWC.14}, \cite{Lee.TWC.2013}.

In this paper, by using tools from stochastic geometry, we aim
at optimizing  bidirectional energy harvesting and information transmission in a
large-scale WPCN.
We consider a new type of dual-function  APs which are able to coordinate energy/information transfer to/from the wireless nodes.
We also consider two types of networks models. As illustrated in  Fig.~\ref{fig: network_model}, in Type-I network model,
the wireless nodes (e.g., the portable electronic devices or the  unmanned vehicles \cite{Type_I}) are assumed to independently
move in the system over frames, while the locations of the APs are fixed.
In Type-II network model,  however, the APs (e.g., the wireless charging vehicles \cite{Type_II}) are assumed to independently
move in the system over frames, while the locations of the wireless nodes are fixed.
 We show that the wireless node's downlink (DL) energy harvesting performance and the uplink (UL) information transfer performance can be
identically characterized for both types of network models.
Moreover, depending on whether each wireless node deploys a rechargeable
battery, we consider two cases with battery-free and battery-deployed wireless
nodes, respectively, and  study the
effects of battery storage.
For both cases, we maximize  the {\it spatial throughput} of the wireless nodes, which is defined as
the total throughput that is achieved  by the wireless nodes per unit network
area averaged over all  information transmission phases (bps/Hz/unit-area) \cite{Che.TWC.14}.

The key contributions of this paper are summarized as follows.
\begin{itemize}
\item \emph{Novel harvest-then-transmit protocol to power a  large-scale network:}
In Section II, we propose a new harvest-then-transmit protocol by extending that in \cite{Ju.TWC},  where
each time frame is partitioned into  a DL phase  for energy transfer from  the APs
 to the wireless nodes, and an UL phase for information transfer
from  each wireless node  to its associated AP. We show that the proposed
harvest-then-transmit protocol is scalable and thus can be efficiently
implemented in a large-scale network.

\item\emph{Problem formulation and simplification for spatial throughput maximization:}
In Section III,  by jointly optimizing
time frame partition between the DL and UL phases and the wireless nodes' transmit
power, we    formulate the spatial throughput
maximization problem under a successful information transmission probability
constraint. To make the problem analytically tractable, we simplify the  problem
  by utilizing the equivalence of the successful information transmission probability
constraint to a transmission probability constraint plus a minimum transmit
power constraint.

\item \emph{Spatial throughput maximization for battery-free wireless nodes:}
In Section IV, we  solve
the spatial throughput maximization problem in the battery-free case, by
studying the
effects of the AP density and the wireless node density. We also show that
at the optimality the wireless nodes  generally prefer to select a small
transmit power, for obtaining large
transmission opportunity.
%Moreover, it is also shown that
%increasing AP density is beneficial for both energy harvesting and information
%transmission. Given the wireless node density, the numbers of optimal
 %time frame partition and UL transmit power solutions are both non-decreasing over
%the   AP density.

\item \emph{Spatial throughput maximization for battery-deployed wireless nodes:}
In Section V, we first study an  ideal infinite-capacity battery scenario, and
show that  all the feasible time frame partition
and  UL transmit power  are optimal, since energy stored in the
battery is sufficient over time.
We then extend our study to the practical finite-capacity battery
scenario. By proposing  a new tight lower bound for the transmission
probability, we approximately solve  the spatial throughput maximization problem.

\end{itemize}

We note only limited studies in \cite{Huang.TWC.14}, \cite{Lee.TWC.2013},  \cite{Huang.IT.13}, and \cite{Dhillon.TWC.14} have adopted
stochastic geometry to  study the large-scale   communication networks enabled by energy harvesting.
Different from these existing studies, we consider the WPCN where dual functional APs  transmit energy and receive information
to/from wireless nodes. Moreover, we focus on characterizing optimal tradeoffs between the DL energy transfer and the UL information transfer,
for both battery-free and battery-deployed cases,
and theoretically analyze the impact of battery storage on
the network throughput performance. In addition, different from most existing studies based on stochastic geometry that only
focused on average system performance of one snapshot, in this paper,  we pursue a long-term average system analysis,
and successfully obtain  tractable system performance in both DL and UL.

\section{System Model}

We consider a WPCN with
stochastically deployed APs and wireless nodes,  where each wireless node harvests
energy broadcast by the APs,  and then uses the harvested energy to support
its information transmission to the associated AP.
As shown in Fig.~1, we assume either   the wireless nodes or the APs move in the system.
  In this section, we  first present  the detailed operations at each wireless node  for
both battery-free and battery-deployed cases, and then
  develop the  network model based on stochastic geometry.

\subsection{System Operation Model}

We consider that  each AP  has
reliable power supply (e.g., by connecting to the power grid or equipping with large-capacity battery storage
in Type-I or Type-II network model, respectively),  while each wireless node  is not
equipped with any embedded energy sources  but  an RF energy harvesting
device.
Thus, the wireless nodes are able to harvest the energy broadcast by the APs,
and use them to support their information transmissions to the APs.
Similar to the practical radio frequency identification (RFID)   system that coexists with
the reader network over the same   frequency (around 915MHz)  \cite{Sample.RFID.07},
we assume all the APs and wireless nodes operate over the
same frequency band.
We  also assume all the APs and wireless nodes are each equipped with a single
antenna, as in the case of the wireless sensor networks.
We  partition  energy transfer and information transfer in
time domain,\footnote{The time-partition-based model
can also be extended to a frequency-partition-based model,
for the wireless devices with multiple antennas and the ability to
operate   over different frequency bands simultaneously as
in \cite{Rui.TWC.13}. Specifically, for a system with
total $T$ frequency bands (like $T$ time slots in this paper), we can assign $N$
bands for energy harvesting and the
remaining $T-N$ bands for information transmission. To optimally decide $N$ and
the UL
transmit power $P_U$, there exists similar tradeoff as  in the
time-partition-based model studied here. } as shown in Fig.~2.
We  assume the network is  frame-based in time and consider a
harvest-then-transmit protocol for the wireless nodes.
Specifically, we assume  each frame  consists of $T>1$
slots, indexing from $0$ to $T-1$, and all the slots are synchronized among APs and wireless nodes.
In each frame, we assign slot $0$ to slot $N-1$, $1\leq
N\leq T-1$, to the APs for broadcasting energy in the DL phase, and assign the
remaining slots, i.e., slot $N$ to slot $T-1$, to the wireless nodes for
transmitting information in the UL phase.
We denote the transmit power of the
APs and the wireless nodes as $P_D>0$ and $P_U>0$, respectively. We assume
$0<P_U\leq P_{\max}$, where $P_{\max}$ is the maximum allowable transmit power
of each wireless node. It is worth noting that to design a \emph{scalable} transmission
scheme for a large-scale WPCN (e.g., wireless sensor or RFID networks),
where the wireless nodes  usually operate at low transmit power,
we consider   the same $P_U$ and $N$ for each wireless node,
and optimize $P_U$ and $N$ globally for a homogeneous stochastic network as will be
shown later. %\footnote{
%In the future work, we are also interested in designing different  $N$ and $P_U$ for
%each wireless mode, though the system-level analysis on the tradeoff
%between energy and information transfer may become  very challenging, due to
%the resultant  non-homogeneous stochastic network in general. }.
Thus, wireless nodes %with time-varying harvested energy
do not need to communicate and coordinate in
 interference management, which is easy to implement in practice.
Moreover, due to the wireless fading channels as well as the low energy
harvesting efficiency of today's RFID technology \cite{EH_system},
 the amount of energy that can be collected in
one slot is usually small, and is difficult to be effectively exploited by the wireless nodes.
As a result, as in the practical energy harvesting devices, e.g., the
  P2110 power harvester receiver \cite{Powercast} designed by the Powercast corporation,
  we consider that a  small-sized capacitor is integrated in the circuits of the
energy harvesting device,\footnote{The integrated capacitor in
the energy harvesting device
  is only used  to improve the energy harvesting efficiency, and thus will not
be exploited   as an energy storage device as the rechargeable battery, which
can manage the harvested energy.}  based on which, the harvested energy from
 slot $0$ to slot $N-1$ in the DL phase can be accumulated without the usage of
additional battery, and
 then entirely boosted out for exploitation by each wireless node (for UL transmission or battery charging), as shown in Fig.~2.
For each wireless node $i$, denote $Z_{F,i}(t)$ as the amount of energy harvested in
DL slot $t$ of frame $F$,
$0\leq t \leq N-1$, $1\leq F\leq \infty$,  and $Z_{F,i}$ as
the total amount of energy harvested in the DL phase of frame $F$.
We  have $Z_{F,i}=\sum_{t=0}^{N-1}Z_{F,i}(t)$.

\begin{figure}
\setlength{\abovecaptionskip}{-0.05in}
\centering
\DeclareGraphicsExtensions{.eps,.mps,.pdf,.jpg,.png}
\DeclareGraphicsRule{*}{eps}{*}{}
\includegraphics[angle=0, width=0.48\textwidth]{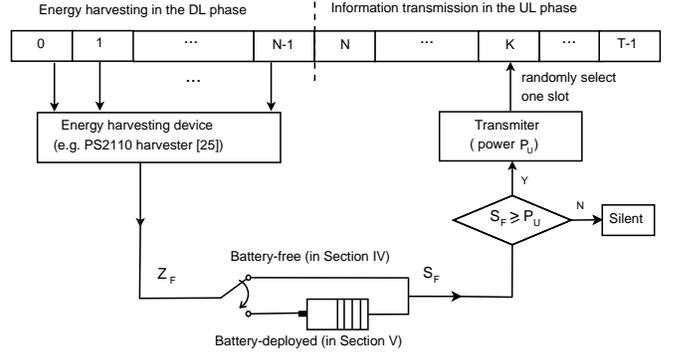}
\caption{Energy harvesting and information transmission for each wireless node in
each frame.}
\label{fig: system_model}
\vspace{-0.12in}
\end{figure}

We  denote $S_{F,i}$ as the amount of energy that is available to wireless
node $i$ at the beginning of  the UL phase of frame $F$.
In the following, depending on whether a wireless node is equipped with a
rechargeable battery
(or any other energy storage devices) to store the total harvested DL energy
$Z_{F,i}$ in each frame
$F$, we consider two cases  with battery-free and battery-deployed wireless nodes, respectively.
In each case,  by applying a $P_U$-threshold
based UL transmission decision as in the literature (e.g., \cite{Huang.TWC.14},  \cite{Lee.TWC.2013}, and \cite{Huang.IT.13}), we
model the evolvement of $S_{F,i}$ over $F$.
For convenience, we assume a normalized unit slot time in the sequel without
loss of generality, and thus we can use the terms of energy and power
interchangeably.

\subsubsection{Battery-free Case}
As show in Fig.~2, in each frame $F$, due to the lack of energy storage, the
wireless nodes manage the harvested energy  in a myopic manner,
i.e., all the harvested energy $Z_{F,i}$ is  consumed within the
current frame $F$. Moreover, if $Z_{F,i}\geq P_U$,
wireless node $i$ decides to transmit information with
 power $P_U$ in the UL phase; otherwise,
it stays silent in the UL phase of frame $F$. Since the unused amount of energy
in the current frame (i.e., $Z_{F,i}-P_U$, if $Z_{F,i}\geq P_U$, or
$Z_{F,i}$, otherwise) will not be kept for future use, for any $F\in\{1,...,\infty\}$, we can easily obtain
\vspace{-0.05in}
\begin{equation}
 S_{F,i}=Z_{F,i}. \label{eq:  S_F_No_Storage}
\end{equation}

\subsubsection{Battery-deployed Case}
Unlike the battery-free case, by deploying a rechargeable battery in the device
circuit,
the wireless nodes can store the unused energy in the current frame for future
use, as long
as the battery capacity  allows.
Thus, the harvested energy can be exploited more effectively in the
battery-deployed case than that in the battery-free case in general.
As shown in Fig.~2, in each frame $F$,  if the battery level at the beginning
of the UL phase,
given by $S_{F,i}$, is no smaller than the required UL transmit power
$P_U$, the wireless node decides to transmit in the UL; otherwise, it stays
silent in the UL phase.
Let the battery capacity be $C$ with $P_U\leq C\leq\infty$.
For any $F\in\{1,...,\infty\}$, given $S_{F-1,i}$, by subtracting the consumed energy in the UL phase of frame $F-1$
and adding the harvested energy in the DL phase of frame $F$, we obtain
 $S_{F,i}$ as \vspace{-0.05in}
\begin{align}
S_{F,i}=\min(S_{F-1,i}-P_UI(S_{F-1,i}\geq P_U)+Z_{F,i}, C), \label{eq:
S_F_Storage}
\end{align}
where $S_{0,i}=0$ and the indicator function  $I(x)=1$ if $x$ is true,
and $I(x)=0$ otherwise. Note that $C=\infty$ is an ideal scenario
with infinite-capacity battery. It is easy to find that in this scenario, for any $F\in\{1,...,\infty\}$,
(\ref{eq: S_F_Storage}) is reduced to
\begin{align}
S_{F,i}=S_{F-1,i}-P_UI(S_{F-1,i}\geq P_U)+Z_{F,i}.  \label{eq:
S_F_Storage_Inf}
\end{align}

At last, in the UL transmission,  for both cases with battery-free and
battery-deployed wireless nodes, we assume there is no transmission coordination
between the wireless nodes for simplicity, as in \cite{Dhillon.TWC.14}.
We thus adopt independent transmission scheduling
for the wireless nodes.\footnote{For simplicity,
we only focus on independent scheduling in this paper.
More advanced scheduling schemes and their effects  in wireless communication networks,
as  in, e.g.,  \cite{Che.TWC.14} and \cite{Wen.scheduling},   will be considered in our future work.}
Specifically,   to reduce the
potentially high interference level in
the UL due to the independent transmissions
of the wireless nodes, we  assume that  if a wireless node $i$
decides to transmit,  it randomly selects a slot from slot $N$ to slot $T-1$ in
the UL phase with equal probability of $1/(T-N)$, and transmits its information
in this slot with transmit power $P_U$ to its nearest AP,  as in
\cite{Huang.TWC.14} and \cite{Andrews.COM.11}, for achieving good communication
quality. The UL information transmission is successful if the received
signal-to-interference-plus-noise-ratio (SINR) at the AP is no smaller than a
target SINR threshold, denoted by $\beta>0$.

\subsection{Network Model}
Based on the operations of the wireless nodes and the APs, in this subsection, we develop the network model
based on stochastic geometry, and then characterize the harvested energy of the wireless node in each frame.

As shown in Fig.~1, we consider two types of network models, which are   Type-I network model, with moving wireless nodes and  static APs, and Type-II network model, with moving APs and static wireless nodes. In both types of networks models, we assume the wireless nodes and the APs are initialized as two independent homogeneous PPPs, denoted by $\Phi(\lambda_w)$, of wireless node density $\lambda_w>0$, and $\Phi(\lambda_{AP})$, of AP density $\lambda_{AP}>0$, respectively. In  Type-I network model, we assume all the APs stay at their initialized locations in all frames, while the wireless nodes independently  change their locations in each frame based on the random walk model considered in \cite{Baccelli.NOW.I}. Specifically,
at the beginning of each frame, each wireless node is independently displaced from its previous location in the proceeding frame  to a new location in the current frame; and stays at its new location  within the current frame. According to the Displacement Theorem in \cite{Baccelli.NOW.I}, the homogeneous PPP $\Phi(\lambda_w)$ is preserved  by the independently displaced wireless nodes in each frame. Similarly, in   Type-II network model, we assume the wireless nodes stay at their initialized locations in all frames, while the APs are independently displaced over frames as the wireless nodes in  Type-I network model. Clearly, the homogeneous PPP  $\Phi(\lambda_{AP})$ is also preserved by the independently displaced APs in each frame in   Type-II network model.

Let $\Phi(\lambda_{AP})=\{X\}$ and $\Phi(\lambda_w)=\{Y\}$, where
$X,Y\in\mathbb{R}^2$
denote the coordinates of the APs and wireless nodes, respectively.
As in  the existing literature that studied wireless charging based on stochastic geometry
(e.g., \cite{Lee.TWC.2013}, \cite{Huang.IT.13}, and \cite{Dhillon.TWC.14}), we assume Rayleigh flat fading channels
with path-loss.\footnote{Since shadowing does not affect the main results of this paper, we ignore the effects of shadowing   for  tractable analysis.}
We also assume the Rayleigh fading channels vary independently over different time slots. In each slot $t$ of a particular frame, the radio signal transmitted by an AP/wireless node is received
at the origin with strength $|X|^{-\alpha}h_X(t)$ and $|Y|^{-\alpha}h_Y(t)$,
respectively, where  $|X|$ and $|Y|$  are the distances from AP $X$ or wireless node $Y$ to the origin $o=(0,0)$, respectively, $h_{X}(t)$ and $h_{Y}(t)$ are independent and identically distributed (i.i.d.) exponential random variables with unit  mean
to model Rayleigh fading in slot $t$ from AP $X$ or wireless node $Y$ to the origin, respectively,
and $\alpha>2$ is the path-loss exponent.

In both Type-I and Type-II network models, due to the stationarity of the homogeneous
PPP $\Phi(\lambda_{AP})$, we focus on a typical wireless node in the DL phase,
which is assumed to be located at the origin, without loss of generality.
For notational simplicity,  for the typical wireless node, we omit
the lowerscript $i$ and use $Z_F(t)$ and $Z_F$ to denote the amount of energy
that is harvested in a particular DL slot $t$ and over all DL slots of frame
$F$, respectively, and use
 $S_F$ to denote the amount of available energy for UL phase in frame $F$.
 Since the harvested energy is obtained from the received RF signals,
as in the existing studies on wireless powered energy harvesting (e.g., \cite{Ju.TWC}, \cite{Rui.TWC.13},  \cite{Huang.TWC.14}, \cite{Lee.TWC.2013} and  \cite{Huang.IT.13}), for any slot $t$ of frame $F$, $0\leq t \leq
N-1$, $1\leq F \leq \infty$, we have
\begin{equation}
 Z_F(t)= \eta \sum_{X\in \Phi(\lambda_{AP})} P_D
|X|^{-\alpha}h_{X}(t), \label{eq: Z_F_t}
\end{equation}
where $\eta \in (0,1)$ is the energy harvesting
efficiency.
As a result, by summing $Z_F(t)$  over all slots in the
DL phase of frame $F$, we obtain \vspace{-0.07in}
\begin{equation}
Z_F= \eta\sum_{X\in \Phi(\lambda_{AP})}
P_D |X|^{-\alpha}\sum_{t=1}^{N}h_{X}(t), \label{eq: Z_F_def}
\end{equation}
where  $\sum_{t=1}^{N}h_{X}(t)$ follows Erlang distribution with shape $N$ and rate $1$. By applying the  PGFL  of the PPP, we obtain the Laplace transform and the complementary cumulative distribution function (CCDF) of $Z_F$ in the following proposition.

\begin{proposition}\label{proposition: Z_F_distribution}
The Laplace transform of $Z_F$ is
\begin{equation}
 \mathcal{L}_{Z_F}\!(s)\!=\!\exp\! \!\left(\!-\pi \lambda_{AP} \frac{\Gamma\!\left(N\!+\!\frac{2}{\alpha}\right)}{\Gamma(N)} \Gamma\!\left(\!1\!-\!\frac{2}{\alpha}\!\right)(P_D \eta s)^{\frac{2}{\alpha}} \!\!\right),
\end{equation}
where $\Gamma(p)=\int_0^{\infty}t^{p-1}e^{-t}\,dt$ is the gamma function. When $\alpha=4$, for any given $z\geq 0$, the CCDF for  $Z_F$ is given  as
\begin{equation}
 \mathbb{P}(Z_F\geq z)= \textrm{erf}\left(\frac{\lambda_{AP}\Gamma\left(N+\frac{2}{\alpha}\right)}{2\Gamma(N)}\sqrt{\frac{\pi^3 P_D \eta}{z}}  \right),  \label{eq: Z_F_CCDF}%\vspace{-0.05in}
\end{equation}
where $\textrm{erf}(x)=\frac{2}{\sqrt{\pi}}\int_0^x \exp(-u^2)\, du$  is the error function.
\end{proposition}

Proposition \ref{proposition: Z_F_distribution} is proved by using a approach  similar to that in \cite{Haenggi.book} for deriving the interference distribution  in a PPP, with the notice that for a random variable $H\sim Erlang(N,1)$, $\mathbb{E}\left(H^m\right)=\frac{\Gamma(N+m)}{\Gamma(N)}$, and thus is omitted here for brevity. It is   clear that Proposition \ref{proposition: Z_F_distribution} holds for both Type-I and Type-II network models.  By increasing $N$ in (\ref{eq: Z_F_CCDF}),  the term $\frac{\Gamma(N+2/\alpha)}{\Gamma(N)}$ increases, and thus the CCDF of $Z_F$ increases for a given $z>0$, as expected.  Moreover,  due to the singularity of the path-loss law $|X|^{-\alpha}$ at the origin,   the average energy
arrival rate is $\mathbb{E}(Z_F)=\infty$. However, this does not necessarily mean that
each wireless node can always harvest sufficient energy, as the probability that
a wireless node can be very close to an AP is very small in any frame in general.
In addition, although from Proposition \ref{proposition: Z_F_distribution}, the distribution of $Z_F$ is identical for each wireless node in each frame $F$, since the harvested energy of each  wireless node   comes from the same
set of APs in $\mathbb{R}^2$ in all frames, it is easy to verify that in both Type-I and Type-II network models, for each wireless node, its harvested energy $Z_F$'s  are not mutually independent over time frames in general;
and for any two wireless nodes locating in different locations in space $\mathbb{R}^2$, their harvested energy are also not mutually independent  in each frame in general.
Since whether a wireless node can transmit in the UL  and the corresponding transmission performance are both strongly depend on the characteristics of $Z_F$'s, similar to the case in \cite{Che.TWC.14}, such correlations  between $Z_F$'s over both time frames and space yield challenges for tractable analysis of the wireless nodes' communication performance as well as the system throughput.

From \cite{Gong.TMC.14},  we observe similar correlation    between  $Z_F$'s over  both time frames and space $\mathbb{R}^2$,  which is determined by the variation of the fading channels as well as  the mobility of the wireless nodes or the APs  in Type-I or Type-II network model, respectively. From (\ref{eq: Z_F_def}), due to the independently varied fading channels between any APs and any wireless nodes over all slots in all frames as well as the independent location change of either the wireless nodes or the APs over frames in the considered models, which
can decorrelate the distance between any APs and any wireless nodes over frames, it is thus easy to verify   that $Z_F$'s correlations over both time frames and space are  weak in general. Moreover, due to the serious path loss for energy transfer and the generally low  energy harvesting efficiency $\eta<1$,  the harvested energy by each wireless node is only dominated by its near  APs. By noticing that the  independent location change of either the wireless nodes or the APs  over time frames  can also  decouple each wireless node's dominated APs over time frames, it is expected that $Z_F$'s correlations over both time frames and space are very weak.
Therefore, to obtain tractable results,  we apply the following independent assumption on the harvested energy $Z_F$'s.
\begin{assumption}
In both Type-I and Type-II network models, $Z_F$'s are mutually independent for each wireless node over frames  and mutually independent for any two different wireless nodes in $\mathbb{R}^2$ in each frame.
\end{assumption}

By Assumption 1, $Z_F$'s become i.i.d. random variables over both time frames and space $\mathbb{R}^2$. We also  successfully validate the feasibility of Assumption 1 later in Section VI-A by simulation.
 In the next section, based on the i.i.d. $Z_F$'s and their identical distribution given in Proposition  \ref{proposition: Z_F_distribution},  we will focus on the system communication metrics in the UL phase, and present the formulation of  the spatial throughput maximization problem.

\section{Performance Metrics and Problem Formulation}
In this section, we focus on studying the information transmission in the UL
phase as system metrics. We first  analyze the point process formed by the wireless nodes that transmit in each slot of the UL phase, and characterize
the   successful information transmission probability of the typical wireless node in the UL.
Then by studying the effects of the design variables $N$ and $P_U$, we
formulate the spatial throughput maximization problem under  a successful information
transmission probability constraint. Since the
successful information transmission probability constraint is very complicated,
we will further simplify it by  finding equivalent constraints, which yields an equivalent spatial
throughput maximization
problem with a simpler structure, as explained later.

\subsection{Successful Information Transmission Probability}
First, we define the \emph{transmission probability}  as  the probability that the typical wireless node can transmit in the UL.
Since $S_F$ is determined by $Z_F$ in both battery-free and battery-deployed cases, given in (\ref{eq:  S_F_No_Storage}) and (\ref{eq: S_F_Storage}), respectively, under Assumption 1 with i.i.d $Z_F$'s over time frames for each wireless nodes, it is easy to verify that   $\{Z_F\}_{1\leq F \leq \infty}$  and thus $\{S_F\}_{1\leq F \leq \infty}$ is ergodic over frame $F$ for both battery-free and battery-deployed cases.  As a result, as in \cite{Huang.IT.13}, since only the wireless nodes  with
$S_{F}\geq P_U$ can transmit to their associated APs,    the transmission probability, denoted by
$\rho$,   is defined as
\begin{equation}
\rho=\lim_{n \to \infty} \frac{1}{n} \sum_{F=1}^{n}\mathbb{P}(S_{F}\geq P_U).
\label{eq: trans_prob_def}
\end{equation}

From Section~II-A,  in both Type-I and Type-II network models, if a wireless node decides to transmit in the UL  based on the transmission probability $\rho$, it randomly selects one slot from the total $T-N$ slots in the UL phase to transmit. Thus, in  the UL phase under both network models, the point process   formed by the  wireless nodes that transmit in each time slot is of the identical active wireless node density,  which is denoted by $\lambda_a$ and given as  \vspace{-0.05in}
\begin{equation}
\setlength{\abovedisplayskip}{7pt}
\setlength{\belowdisplayskip}{7pt}
 \lambda_{a}=\frac{\lambda_w \rho}{T-N}. \label{eq: lambda_a}
\end{equation}
Due to the correlated harvested energy for different wireless nodes in each frame, as discussed in Section II-B, the point process formed by the active wireless nodes in each UL slot is not a PPP in general. However,  \emph{as a direct result by applying Assumption 1 with i.i.d. $Z_F$'s for different wireless nodes in each frame, the active wireless nodes' transmissions in each UL slot become independent, which yields a homogeneous PPP   for each UL slot, denoted by $\Phi(\lambda_a)$, of identical density $\lambda_{a}$, in both Type-I and Type-II network models.}   We also successfully validate such PPP assumption in the UL slot later in Section VI-A by simulation.

Next, since each active wireless node  only selects one slot in the UL phase to transmit,
we focus on a particular slot in the UL, and  analyze the typical wireless node's  information transmission performance in the UL based on the PPP $\Phi(\lambda_a)$, under both Type-I and Type-II network models.
 Similar to   the DL phase studied in Section II-B, due to the stationarity of
$\Phi(\lambda_{a})$, we assume the typical  wireless node's
associated AP  is located at the origin in the UL phase, without loss
of generality. For ease of notation, we omit the time slot index  and use $h_m$ to denote   the Rayleigh fading channel from the typical wireless node locating at $m \in \mathbb{R}^2$ to the origin. Suppose  $|m|=r$ is the random distance between the typical wireless node and its
associated AP. Let  $\sigma^2>0$ be the noise power. We then define the {\it
successful information transmission probability} as $P_{suc}$, which gives the
probability that the
received SINR   at the typical AP is no smaller than the target level $\beta$
and can be
written as
\begin{align}
 P_{\textrm{suc}}=\mathbb{P}\left( \frac{P_U h_m
r^{-\alpha}}{\sum_{Y\in \Phi(\lambda_{a}),Y\neq m} P_U h_{Y}
|Y|^{-\alpha}+\sigma^2}\geq \beta \right). \label{eq: P_suc_def}
\end{align}
Since the typical wireless node is
associated with its nearest AP, i.e., no other APs can be closer
than $r$, the probability density function (pdf) of $r$ can be easily
found by using the null probability of a PPP, which is given by $ f_r(r)=2\pi
\lambda_{AP} r e^{-\lambda_{AP} \pi r^2}$ for $r\geq0$ \cite{Andrews.COM.11}.

Given a generic transmission probability $\rho$, we derive the expression
of $P_{\textrm{suc}}$,
defined in (\ref{eq: P_suc_def}).
By applying the PGFL of a PPP \cite{Stoyan.SG.95},
we  explicitly express  $P_{\textrm{suc}}$
for a given $\rho$  in the following proposition.
\begin{proposition} \label{proposition: P_suc}
Given the transmission probability $\rho$,
the successful information transmission probability for the typical wireless node
is %\vspace{-0.05in}
\begin{equation}
\setlength{\abovedisplayskip}{7pt}
\setlength{\belowdisplayskip}{7pt}
 P_{suc}=\pi \lambda_{AP} \int_0^{\infty} e^{-ax} e^{-bx^{\frac{\alpha}{2}}}
\,dx,
\label{eq: P_suc_general} %\vspace{-0.05in}
\end{equation}
where $a=\pi\lambda_{a}\kappa+\pi\lambda_{AP}=\frac{\pi \kappa \lambda_w
\rho}{T-N}+\pi \lambda_{AP}$, with
$\kappa=\beta^{\frac{2}{\alpha}}\int_0^{\infty}\frac{1}{1+u^{\frac{\alpha}{2}}}
\,du$, and $b=\frac{\beta \sigma^2}{P_U}$. When $\alpha=4$,
(\ref{eq: P_suc_general}) admits a closed-form expression with %\vspace{-0.05in}
\begin{equation}
\setlength{\abovedisplayskip}{7pt}
\setlength{\belowdisplayskip}{7pt}
   P_{suc}= G\exp\left(\Upsilon^2/2 \right) Q(\Upsilon), \label{eq:
P_suc_4} %\vspace{-0.05in}
\end{equation}
where $G=\pi^{\frac{3}{2}} (\beta \sigma^2)^{-\frac{1}{2}}\lambda_{AP}
\sqrt{P_U}$,  $\Upsilon=\frac{G}{\sqrt{2 \pi}}+\frac{\pi^2\sqrt{P_U}
\lambda_w\rho}{2(T-N)\sqrt{2 \sigma^2}}$, and
$Q(x)=\frac{1}{\sqrt{2 \pi}}\int_x^{\infty}\exp\left(-\frac{u^2}{2} \right)\,
du$ is the standard Gaussian tail probability.
\end{proposition}

Proposition \ref{proposition: P_suc} is proved using a method similar to that for proving  Theorem 2 in \cite{Andrews.COM.11}, and thus is omitted for brevity.
Clearly, Proposition \ref{proposition: P_suc} also holds for both Type-I and Type-II network models as Proposition \ref{proposition: Z_F_distribution}.
It is observed from  both (\ref{eq: P_suc_general}), for a general $\alpha$,
and (\ref{eq: P_suc_4}), for $\alpha=4$,  by
decreasing the transmission probability $\rho$, due to the reduced active wireless node density $\lambda_a$,
given in (\ref{eq: lambda_a}), the interference level in the UL phase is
reduced, and thus $P_{suc}$ is increased.

In the next subsection, by applying   Proposition  \ref{proposition: P_suc}, we
formulate the spatial throughput maximization problem.
It is worth noting that since identical  DL and UL performance  are obtained for Type-I and Type-II network models,  same spatial throughput maximization  problem formulation and corresponding solutions are obtained for the two models, and thus we will not differentiate the two models in the sequel of this paper.
\vspace{-0.05in}

\subsection{Spatial Throughput Maximization Problem}
We  focus on the effects of   the number of slots $N$
assigned to the DL  phase  and the UL transmit power $P_U$ to investigate the
interesting tradeoff between the energy transfer in
the DL and the information transfer in the UL.
By increasing $N$ at a fixed $P_U$, from (\ref{eq: Z_F_def}),
we observe that the harvested energy
$Z_F$ in the DL increases, and thus  the transmission probability
$\rho$ is increased. As a result, the successful information transmission probability
 $P_{suc}$ in the UL, given in (\ref{eq: P_suc_general}) for a general $\alpha$ or
(\ref{eq: P_suc_4}) for $\alpha=4$,  is decreased.
Similarly, by increasing $P_U$ at a
fixed $N$, we observe  a decreased transmission probability $\rho$ in (\ref{eq:
trans_prob_def}), and thus an increased   $P_{suc}$ in the UL.
In the following, we  design   $N$ and  $P_U$ to optimize the
network performance.

Specifically,  to ensure the QoS for each wireless node, we apply a  \emph{successful information
 transmission probability constraint} such that $P_{suc}\!\geq \!1-\epsilon$, with
$\epsilon \!\ll \!1$, for any    DL slot allocation $N$ and UL transmit power
$P_U$.
Similar to \cite{Huang.IT.13}, we define the \emph{spatial throughput} of the considered WPCN as
the total throughput that is achieved by the wireless nodes over  all
the slots in the UL phase per unit network area  (bps/Hz/unit-area).
Moreover, given the receiver SINR threshold $\beta$,
we suppose  the uplink information transmission is successful,  if  the   information can be coded at a rate $\log_2(1+\beta^{\tau})$
with $\tau\!\geq \!1$. Assuming $\tau\!=\!1$, the spatial throughput is then given by
\begin{align}
R(N,P_U)&=\lambda_{a}(T-N)    \log_2(1+\beta)   \nonumber \\
&\overset{(a)}{=} \lambda_w \rho \log_2(1+\beta), \label{eq:
spatial_throughput_def}
\end{align}
where procedure $(a)$  is obtained by  applying   (\ref{eq: lambda_a}).
To be precise, $R(N,P_U)$ should be scaled by the successful information transmission probability
$P_{suc}$;  but since $P_{suc}$ is ensured to be very close to
$1$ given $\epsilon\ll 1$, this factor is omitted for ease of  notation as
in \cite{Lee.TWC.2013} and \cite{Huang.IT.13}.
It is also easy to find that due to $\rho$ defined in (\ref{eq: trans_prob_def}),  $R(N,P_U)$ is a function
of $N$ and $P_U$.
Moreover, in each frame consisting of $T$ slots, since we should at least
assign one slot to UL phase for information transmission, we have $N\leq T-1$.
Hence, under the successful
information transmission probability constraint, we formulate the
spatial throughput maximization problem as
\begin{align}
\textrm{(P1)}:~~~\mathop{\mathrm{max.}}_{N, P_U}&~~ R(N,P_U) \nonumber \\
\mathrm{s.t.} & ~~  P_{suc}\geq 1-\epsilon, \nonumber \\
& ~~ N\in\{1,...,T-1\}, \nonumber \\
& ~~ 0 < P_U \leq P_{\max}. \nonumber
\end{align}

It is noted that Problem (P1) involves integer programming, due to
$N\in\{1,...,T-1\}$.
Moreover,  since   the  expression of $P_{suc}$   in Proposition
\ref{proposition: P_suc} is very complicated, it is  difficult to analyze the
effects of $N$ and $P_U$ to   ensure $P_{suc}\geq 1-\epsilon$, and  thus solve
Problem (P1). In the following proposition for  the case of $\alpha=4$, which is
a typical channel fading exponent in wireless communications,    we successfully
find the equivalent constraints to $P_{suc}\geq 1-\epsilon$, which can be used
for formulating an equivalent problem to Problem (P1) with a simpler structure.

 \begin{proposition} \label{proposition: P_suc_constraint_Equivalent}
When  $\alpha=4$, as  $\epsilon \rightarrow 0$, the successful
information transmission probability
 constraint $P_{suc}\geq 1-\epsilon$  is equivalent to a transmission
probability constraint $\lambda_w \rho
\leq K_{\epsilon}\lambda_{AP} (T-N)$ with $P_U\geq \frac{g_0^2 \beta
\sigma^2}{\pi^3 \lambda_{AP}^2}$ and
$K_{\epsilon}=\frac{2\epsilon}{1-\epsilon}\frac{\beta^{-\frac{1}{2}}}{\pi}$,
where
$g_0$ is the unique solution to
$g_0 Q\left(\frac{g_0}{2 \pi}\right) =(1-\epsilon)\exp \left(-\frac{g_0^2}{4
\pi}\right)$.
 \end{proposition}
 \begin{proof}
   Please refer to Appendix \ref{appendix: proof_euqivalent_trans_prob}.
 \end{proof}

 \begin{remark}
 Since we assume $\epsilon\ll 1$  to assure  satisfied QoS in the UL
transmission,
 Proposition \ref{proposition: P_suc_constraint_Equivalent} can be well applied
 in our considered system.  Moreover, the noise power $\sigma^2\neq 0$ provides
a valid minimum transmit power level   for $P_U\leq P_{\max}$, denoted by
$P_{\min}=\frac{g_0^2 \beta \sigma^2}{\pi^3 \lambda_{AP}^2}$, which is
important to assure  a sufficiently large $P_{suc}$ in a noise-dominant
network.
To avoid the trivial case without any valid decision  for $P_U$,
we assume $P_{\min}\leq P_{\max}$.
 \end{remark}

Finally, for ease of analysis, we focus on the case of
$\alpha=4$ in the sequel.\footnote{The value of $\alpha$ does not
affect the main results of this paper.  Moreover, for other cases with
$\alpha\neq 4$, the spatial
throughput maximization problem can be similarly studied by using the
modeling methods provided in
this paper.}  By applying Proposition \ref{proposition:
P_suc_constraint_Equivalent},  we find an equivalent
problem to Problem (P1), which is given by
\vspace{-0.02in}
\begin{align}
\textrm{(P2)}:~~~\mathop{\mathrm{max.}}_{N, P_U}&~~ R(N,P_U) \nonumber \\
\mathrm{s.t.} & ~~  \lambda_w \rho
\leq K_{\epsilon}\lambda_{AP} (T-N),
\nonumber \\
& ~~ N\in\{1,...,T-1\}, \nonumber \\
& ~~ P_{\min}\leq P_U \leq P_{\max}. \nonumber
\end{align}
Clearly, Problem (P2) has transformed the successful information transmission
probability constraint in (P1) to an equivalent transmission probability
constraint on $\rho$. In the next two
sections, we solve Problem (P2) for
both battery-free and battery-deployed cases,  and study the effects of battery
storage on the achievable throughput. \vspace{-0.07in}

\section{Wireless Powered Information Transmission in  Battery-Free Case}
Due to the limited circuit size  of some wireless devices,  it is hard to
install a sizable battery for these devices to store the harvested energy. Thus, the use
of battery-free wireless devices  is growing in many wireless applications
(e.g., the body-worn sensors for health monitoring).
In this section, we focus on the spatial throughput maximization problem for the
battery-free wireless nodes.
We  first derive the transmission probability $\rho$ and the spatial throughput
$R(N,P_U)$. We then   substitute  $\rho$ and $R(N,P_U)$ into Problem
(P2) and solve the spatial throughput maximization problem, by finding the
optimal solution of  $N^{*}$  and $P_U^{*}$.

First, we derive the transmission probability $\rho$.  In the battery-free  case,
as introduced in Section II-A, the wireless node operates with its available
energy according to (\ref{eq: S_F_No_Storage}). Thus, by substituting (\ref{eq:
S_F_No_Storage})  into (\ref{eq: trans_prob_def}), we obtain the
expression of $\rho$ in the battery-free case as \vspace{-0.05in}
\begin{align}
 \rho&=\lim_{n \to \infty} \frac{1}{n} \sum_{F=1}^{n}\mathbb{P}(Z_{F}\geq P_U)
\overset{(a)}{=}\mathbb{P}(Z_{F}\geq P_U) \nonumber \\
& \overset{(b)}{=} \textrm{erf}\left(\frac{\Gamma(N+2/\alpha) \lambda_{AP}}{2\Gamma(N)}\sqrt{\frac{\pi^3 P_D\eta}{P_U}} \right), \label{eq: rho_No_storage}
\end{align}
where procedure $(a)$ follows from our assumption in Section II-B, which gives i.i.d. $Z_F$'s  for the typical wireless node over frames, and procedure $(b)$ follows from
(\ref{eq: Z_F_CCDF}), by replacing $z$ with
$P_U$.  Note that since the error function $\textrm{erf}(x)\rightarrow 1$ when
$x$ is sufficiently large,
from (\ref{eq: rho_No_storage}), $\rho\rightarrow 1$, by adopting sufficiently
large $N$, $\lambda_{AP}$, and/or $P_D$. By substituting (\ref{eq:
rho_No_storage}) into (\ref{eq: spatial_throughput_def}),
we obtain the spatial throughput for the battery-free case as
%\vspace{-0.05in}
\begin{align}
R(N,P_U)\!=\!\lambda_w\textrm{erf}\left(\!\frac{\Gamma(N\!+\!2/\alpha) \lambda_{AP}}{2\Gamma(N)}\sqrt{\frac{\pi^3 P_D\eta}{P_U}} \!\right)\! \log_2(1\!+\!\beta).
\label{eq: throughput_batter_free}
\end{align}
In addition,  by substituting (\ref{eq: rho_No_storage}) into (\ref{eq:
P_suc_4}),  the expression of  $P_{suc}$ in the battery-free case can also be
easily obtained.

Next, by substituting $\rho$, given by (\ref{eq: rho_No_storage}),  and
$R(N,P_U)$, given by (\ref{eq: throughput_batter_free}),   into Problem  (P2),
we obtain the spatial throughput
maximization problem for the battery-free case as
\begin{align}
\textrm{(P3)}\!:\mathop{\mathrm{max.}}_{N,
P_U}&~\!\lambda_w\textrm{erf}\!\left(\!\frac{\Gamma(N\!+\!\frac{2}{\alpha}) \lambda_{AP}}{2\Gamma(N)}\sqrt{\frac{\pi^3 P_D\eta}{P_U}}\! \right)\! \log_2(1\!+\!\beta) \nonumber \\
\mathrm{s.t.} &~\!   \frac{\lambda_w}{K_{\epsilon}} \textrm{erf}\!\left(\!\frac{\Gamma(N\!+\!\frac{2}{\alpha}) \lambda_{AP}}{2\Gamma(N)}\sqrt{\frac{\pi^3 P_D\eta}{P_U}} \!\right)\!\!\leq \!\lambda_{AP} (T\!\!-\!\!N), \nonumber \\
\label{eq: rho_constraint_BF}\\
& ~N\in\{1,...,T-1\}, \nonumber \\
& ~ P_{\min}\leq P_U \leq P_{\max}.
\nonumber
\end{align}

It is observed that in Problem (P3), both the objective function and  the
transmission probability constraint, given by (\ref{eq: rho_constraint_BF}),
are related to the error function.
Note that the error function $\textrm{erf}(x)$  increases over $x \geq0$, and
then converges to its maximum value $1$ when $x$ is sufficiently large. Suppose
at  $x=v_e$, we have $1-\textrm{erf}(v_e)=10^{-n}$, where we assume
$n>0$ is sufficiently large such that when $x\geq v_e$, $\textrm{erf}(x) = 1$
holds with an ignorable absolute error, which is no larger than $10^{-n}$.
Under such a tight approximation, to help solve Problem (P3), we  calculate
$\textrm{erf}(x)$ over $x\geq 0$ as:
\begin{equation}
  \textrm{erf}(x)=\left\{
   \begin{array}{ll}
   \textrm{erf}(x),  &\textrm{~if~}x<v_e\\
   1,  &\textrm{~if~}x\geq v_e.
   \end{array}
  \right.  \label{eq: erf}
\end{equation}

It is also observed that in Problem (P3), the maximum achievable spatial
throughput over all $N\in\{1,...,T-1\}$
and $P_U\in[P_{\min},P_{\max}]$  is $\lambda_w \log_2(1+\beta)$, and it is
achieved when the transmission
probability $\rho=1$, i.e.,
\begin{equation}
\frac{\Gamma(N+2/\alpha) \lambda_{AP}}{2\Gamma(N)}\sqrt{\frac{\pi^3 P_D\eta}{P_U}} \geq v_e, \label{eq:
rho_1_v_e}
\end{equation}
by applying (\ref{eq: erf}) to (\ref{eq: rho_No_storage}).
Moreover, for any given wireless node density $\lambda_w>0$, if
the AP density $\lambda_{AP}$ is sufficiently large, such that
$K_{\epsilon}\lambda_{AP} (T-N)\geq \lambda_w$ holds
for any $N\in\{1,...,T-1\}$, the transmission probability constraint given in
(\ref{eq: rho_constraint_BF}) of
Problem (P3) is always satisfied, and thus any  $N\in\{1,...,T-1\}$ and
$P_U\in[P_{\min},P_{\max}]$ that satisfy (\ref{eq: rho_1_v_e})   is optimal to
Problem (P3).
However, if $\lambda_{AP}$ is too small,  the transmission probability
constraint in (\ref{eq: rho_constraint_BF})
may not be able to be satisfied with
any $N\in\{1,...,T-1\}$ and $P_U\in[P_{\min},P_{\max}]$, and thus Problem (P3)
has no solution.
Therefore, in the following theorem, by taking the wireless node density as a
reference,  we divide the   AP density into three regimes, each with
different optimal solutions to Problem (P3), and present for these  optimal solutions
the resulting  maximized spatial throughput in each regime.

\begin{theorem}\label{theorem: P-NB}
In the battery-free case, the optimal solutions $N^{*}$ and $P_U^{*}$ to Problem (P3)
are determined as follows, where in each AP density regime,  the
corresponding maximum spatial throughput $R(N^{*},P_U^{*})$ is obtained by substituting
the optimal $N^{*}$ and $P_U^{*}$ to (\ref{eq: throughput_batter_free}).
\begin{enumerate}
\item In the {\it high AP density  regime} ($\lambda_{AP} \geq \frac{\lambda_w}{K_{\epsilon}}$),   the transmission probability constraint
given in (\ref{eq: rho_constraint_BF}) is always satisfied. The optimal solutions are given by
\begin{equation}
\!  \left\{\!\!\!
   \begin{array}{l}
    \forall N^{*}\!\!\in \!\{1,...,T\!-\!1\} \text{~and~} \forall P_U^{*}\!\in \![P_{\min}, P_{\max}] \\
    \text{that~ satisfy~(\ref{eq: rho_1_v_e}),} ~~~~~~~\textit{if~}\text{(\ref{eq: rho_1_v_e})~holds~when~}  \\
    ~~~~~~~~~~~~~~~~~~~~~~~~~~~N\!=\!T\!-\!1,P_U\!=\!P_{\min},\\%\text{~i.e.,~}  \rho=1 \\
    N^{*}\!\!=\!T\!-\!1, P_U^{*}\!=\!P_{\min},\textit{otherwise}.
   \end{array}
  \right.
\end{equation}

\item In the {\it medium AP density   regime}
($\frac{\lambda_w}{K_{\epsilon}(T-1)}\leq \lambda_{AP}
<\frac{\lambda_w}{K_{\epsilon}}$),  a unique $N_0\in \{1,...,T-2\}$ exists
such that $K_{\epsilon}\lambda_{AP}(T-(N_0+1))<\lambda_w\leq
K_{\epsilon}\lambda_{AP}(T-N_0)$. Thus,   (\ref{eq:
rho_constraint_BF}) is always satisfied when $N\leq N_0$. The optimal solutions are then given by
\begin{equation}
\!  \left\{\!\!\!
   \begin{array}{l}
    \forall N^{*}\!\in \!\{1,...,N_0\} \text{~and~} \forall P_U^{*}\!\in \![P_{\min}, P_{\max}] \\
    \text{that~ satisfy~(\ref{eq: rho_1_v_e}),} ~~~~~~~~~\textit{if~}\text{(\ref{eq: rho_1_v_e})~holds~when~}  \\
    ~~~~~~~~~~~~~~~~~~~~~~~~~~~~~N\!=\!N_0,P_U\!=\!P_{\min},\\%\text{~i.e.,~}  \rho=1 \\
    N^{*}=N_0, P_U^{*}=P_{min},~\textit{otherwise}.
   \end{array}
  \right.
\end{equation}

\item In the {\it low AP density regime}
($\lambda_{AP}<\frac{\lambda_w}{K_{\epsilon}(T-1)}$), we find (\ref{eq:
rho_constraint_BF}) cannot be satisfied with $\rho=1$. We thus obtain the following.
\begin{itemize}
 \item \emph{if} $ \lambda_w \textrm{erf}\left(\frac{\Gamma(N+2/\alpha) \lambda_{AP}}{2\Gamma(N)}\sqrt{\frac{\pi^3 P_D\eta}{P_{max}}} \right) \!> \!K_{\epsilon}\lambda_{AP} (T\!-\!N)$
at $N=1$,  no feasible solutions exist;
\item \emph{otherwise}, $N^{*}$ and $P_U^{*}$ are obtained by Algorithm~1, where $P_{s}$ in Line 4 is the unique solution to $\lambda_w
\textrm{erf}\left(\frac{\Gamma(N+2/\alpha) \lambda_{AP}}{2\Gamma(N)}\sqrt{\frac{\pi^3 P_D\eta}{P_U}} \right) \!=\! K_{\epsilon}\lambda_{AP} (T\!-\!N)$, and $\textrm{erfinv}(x)$ is the
inverse error function of $x\geq 0$.
\end{itemize}

\end{enumerate}
\end{theorem}
\begin{proof}
 Please refer to Appendix \ref{appendix: proof_theorem}.
\end{proof}

\begin{remark}
For the battery-free case,
due to the lack of energy storage, the amount of available  energy for the UL
phase in each frame
 is strongly affected by the time-varying DL channel fading, and thus may often
be of a small value. As a result, we observe from Theorem  \ref{theorem: P-NB}
that  \emph{to obtain more opportunity to transmit in the UL, the wireless nodes prefer
to set $P_U=P_{\min}$}. Moreover, given the wireless node density $\lambda_w$,  by increasing the AP density $\lambda_{AP}$,
we observe   {\it double} performance improving  effects in the WPCN:
1) in the DL phase, the amount of harvested energy at each wireless node in the DL
phase increases over $\lambda_{AP}$;
2) in the UL phase, due to the largely shortened distance between each wireless node and its associated AP by increasing $\lambda_{AP}$,
  the desired signal strength at the AP is substantially increased, which
dominates over the increased interference effects in the UL.%\footnote{Similar effects of the AP density on improving the wireless node's communication quality  in the UL have also
%been studied in  \cite{Andrews.COM.11}.}.
We thus  find the resulting successful information
transmission probability in the UL phase is increased.
  As a result,  the successful information transmission probability
constraint becomes loose by adopting a large AP density; and
thus from Theorem  \ref{theorem: P-NB}, both the number of   optimal
solutions  and the maximized spatial throughput
are non-decreasing over $\lambda_{AP}$.
\end{remark}

\begin{algorithm}
\caption{Efficient algorithm for optimally solving Problem (P3) in the low AP density regime.}
\begin{algorithmic}[1]
\STATE initialize $N^{*}=0$, $P_U^{*}=0$, and $R(N^{*},P_U^{*})=0$.
\FOR{each $N\in\{1,...,T-1\}$}
\IF{$\lambda_w \textrm{erf}\left(\frac{\Gamma(N+2/\alpha) \lambda_{AP}}{2\Gamma(N)}\sqrt{\frac{\pi^3 P_D\eta}{P_U}} \right) \leq K_{\epsilon}\lambda_{AP} (T-N)$}
\STATE set $P_{s}\!=\!\pi^3 P_D\eta \! \left[\!\frac{2\Gamma(N)}{\lambda_{AP}\Gamma(N\!+\frac{2}{\alpha})}
\textrm{erfinv}\left(\! \frac{K_{\epsilon}\lambda_{AP}(T\!-\!N)}{\lambda_w}\!\right)
\!\right]^{\!-2}$.
\IF{$P_{s}<P_{\min}$}
\STATE set $p=P_{\min}$;
\ELSE
\STATE set $p=P_{s}$;
\ENDIF
\IF{$R(N,p)>R(N^{*},P_U^{*})$}
\STATE set $P_U^{*}=p$, $N^{*}=N$, and $R(N^{*},P_U^{*})=R(N,p)$.
\ENDIF
\ENDIF
\ENDFOR
\STATE return $N^{*}$, $P_U^{*}$, and $R(N^{*},P_U^{*})$.
\end{algorithmic}
\end{algorithm}

\vspace{-0.1in}

\section{Wireless Powered Information Transmission  in   Battery-Deployed Case}
In this section, we consider the  case with battery-deployed wireless nodes, as
shown in Fig.~2 and discussed in Section II-A.
In the following, we first study the ideal scenario with
infinite-capacity battery, i.e., $C=\infty$, to help understand the effects of
deploying a battery for improving the  network performance. Then, we focus on
a more practical scenario  with a finite-capacity battery, i.e., $C<\infty$. One
can imagine that  the network performance of the scenario with $C<\infty$ is
upper bounded by that  with $C=\infty$.
\vspace{-0.03in}

\subsection{Infinite-Capacity Battery Scenario ($C=\infty$)}
In this subsection, we consider the ideal scenario with $C=\infty$, for which
the battery level $S_F$ evolves over frames  according to (\ref{eq:
S_F_Storage_Inf}).
In the following, we  derive the transmission probability $\rho$ and the spatial
throughput $R(N,P_U)$ in this scenario. Then by substituting
$\rho$ and $R(N,P_U)$ into  Problem (P2),
we provide the optimal solutions
$N^{*}$ and $P_U^{*}$ for the infinite-capacity scenario.

First, we study the  transmission probability $\rho$. Unlike the battery-free
case, by deploying batteries to store the harvested
energy over frames, the time-varying channel effects on the available amount
of energy in the UL phase  is largely alleviated in the battery-deployed case.
When $C=\infty$, all the harvested energy $Z_F$ in each frame can
  be stored in the battery and used by the following frames.
  Moreover, note that
the average energy arrival rate  in the DL phase of
each frame  is $\mathbb{E}(Z_F)=\infty$, as explained in Section II-B, which is
much larger than the
required transmit power $P_U<\infty$    in the UL phase. Thus, as the
harvested energy accumulates in the battery over frames, we obtain the following proposition.
\begin{proposition} \label{proposition: rho_storage_Inf}
Given infinite-capacity battery, the typical node's UL transmission probability
is \vspace{-0.07in}
\begin{equation}
 \rho=1. \label{eq: rho_storage_Inf} \vspace{-0.07in}
\end{equation}
\end{proposition}
\begin{proof}
Please refer to Appendix \ref{appendix: proof_rho_inf_B}.
\end{proof}

From Proposition \ref{proposition: rho_storage_Inf},
in the  infinite-capacity scenario, the APs' RF signals can be considered as a
\emph{reliable} energy source for the wireless nodes.

Next, by substituting (\ref{eq: rho_storage_Inf}) into
(\ref{eq: spatial_throughput_def}), which is the objective function in Problem
(P2), we   obtain the spatial throughput in the infinite-capacity
battery case as $R(N,P_U)=\lambda_w\log_2(1+\beta)$, which is a constant and is
the maximum achievable
spatial throughput of $R(N,P_U)$. Since $R(N,P_U)$ is a constant, the objective
function in
Problem (P2) is also a constant. Thus, the spatial throughput maximization
problem degenerates to a feasibility problem,  given by
\begin{align}
\textrm{(P4)}~~~~~\mathrm{Find}&~~ N, P_U
\nonumber \\
\mathrm{such~ that}& ~~N\!\in \!\left \{\!1,..., \min\!\left(\!T\!-\!1,
T\!-\!\frac{\lambda_w}{K_{\epsilon}\lambda_{AP}}\!\right)\right\}, \label{eq: S_N} \\
& ~~ P_{\min}\leq P_U \leq P_{\max}.
\label{eq: S_P}
\end{align}

By observing the constraints in Problem (P4), the optimal solutions $N^{*}$ and
$P_U^{*}$  of the spatial throughput maximization problem
 are given by any arbitrary point   in the rectangular feasible
region defined by (\ref{eq: S_N}) and (\ref{eq: S_P}).  Unlike the optimal
solutions in the battery-free case,
given in Theorem \ref{theorem: P-NB}, where  $N^{*}$ and
$P_U^{*}$  are correlated, $N^{*}$ and $P_U^{*}$ here can be independently
selected for the   infinite-capacity battery case.
Moreover,  since the transmission probability $\rho=1$, the wireless nodes
can always have opportunity to transmit in the UL with sufficient energy. Thus,
we find any transmit power level
$P_U\in[P_{\min}, P_{\max}]$ is optimal in the infinite-capacity battery
scenario. This is  in  sharp contrast
to the battery-free case, where the optimal transmit power level is the minimum
 transmit power $P_{\min}$ in general. However,  similar to the
battery-free case results shown in Theorem \ref{theorem: P-NB},
since the number of feasible $N$'s is non-decreasing over $\lambda_{AP}$
from (\ref{eq: S_N}),
the number of optimal solution pairs ($N^{*}$ and $P_U^{*}$)
is non-decreasing over $\lambda_{AP}$.
Moreover, we note that if $\lambda_{AP}\geq
\frac{\lambda_w}{K_{\epsilon}}$,  which is the high AP density regime defined
for the battery-free case in Theorem \ref{theorem: P-NB}, any $N \in
\{1,...,T-1\}$ is optimal for the infinite-capacity battery case. This is
because the transmission probability constraint given in Problem (P2) is
always satisfied and  the UL transmission interference is small due to the close
distance  between each wireless node  and its associated AP.
\vspace{-0.05in}

\subsection{Finite-Capacity Battery Scenario ($C<\infty$)}
In this subsection, we consider a practical scenario with finite-capacity
battery, i.e., $C<\infty$, in which the network performance is upper
and lower bounded by that in the infinite-capacity battery scenario and
battery-free case, respectively.
Since the stored energy is capped by $C$, the battery level evolution,
given in (\ref{eq: S_F_Storage}), and thus
 the transmission probability $\rho$, defined in (\ref{eq: trans_prob_def}),
are all  dependent on $C$.
It is hence difficult to find an exact expression of $\rho$ for the
finite-capacity battery scenario \cite{Huang.IT.13}. As a result, we focus on
providing effective bounds to $\rho$.
In the following, we first provide closed-form  lower and
 upper bounds  of $\rho$,   based on
 which,  a special case with $\rho=1$ is obtained.
 Since the tightness of these  closed-form bounds cannot be assured, we then
provide another lower bound, which is relatively tighter to $\rho$ but can only
be obtained numerically. At last,  by applying the obtained bounds of $\rho$, we
study the spatial throughput maximization problem for the finite-capacity
battery scenario.

\subsubsection{Closed-form Bounds of Transmission Probability $\rho$}
By noticing from (\ref{eq: S_F_Storage}) and (\ref{eq: trans_prob_def}), the
transmission probability $ \rho$ increases over the battery capacity $C$.
We thus find  the transmission probability in the finite-capacity battery scenario
is upper and lower bounded by that in the infinite-capacity battery scenario,
given in (\ref{eq: rho_storage_Inf}), and that in the battery-free case,
given in (\ref{eq: rho_No_storage}), respectively.  However, it is noted that both (\ref{eq:
rho_No_storage}) and (\ref{eq: rho_storage_Inf}) are constants and thus cannot
flexibly capture the variation of $\rho$  over different values of  capacity
$C$. It is also noted that \cite{Huang.IT.13} has proposed a
 lower bound, $1-e^{-Q(C-P_U)}$, where $Q$ is the root of $\ln
\mathbb{E}\left[e^{-Q(Z_F-P_U)} \right]$,
  under the condition that $\mathbb{E}(Z_F)>P_U$. Although such a lower
bound exponentially increases over $C$ and can also  be applied in our
considered system as  $\mathbb{E}(Z_F)=\infty$,   it  may not be tight when $C$
is small.
For example, when $C=P_U$, the lower
bound $1-e^{-Q(C-P_U)}$  provided in \cite{Huang.IT.13} is $0$, which is even
smaller than the
lower bound   given in (\ref{eq: rho_No_storage}).
As a result, we combine both lower bounds given in  (\ref{eq: rho_No_storage})
and \cite{Huang.IT.13} to provide a tighter lower bound in the  following
proposition.
 \begin{proposition} \label{proposition: rho_bound_closedForm}
For the finite-capacity battery case, the transmission probability
$\rho$ satisfies $\mathcal{L}\leq \rho \leq 1$, where
$\mathcal{L}\!=\!\max\left(\textrm{erf}\left(\!\frac{\Gamma(N\!+\!2/\alpha) \lambda_{AP}}{2\Gamma(N)}\sqrt{\frac{\pi^3 P_D\eta}{P_U}}\! \right), 1\!-\!e^{-Q(C-P_U)} \right)$, with
$Q=P_D \eta \left[\frac{\pi \lambda_{AP}\Gamma(N+2/\alpha)\Gamma\left(1-2/\alpha \right)}{P_U\Gamma(N)}\right]^2$.
\end{proposition}
\begin{proof}
 Please refer to Appendix \ref{appendix: proof_rho_finite_B}.
\end{proof}

It is noted that  when
$\lambda_{AP}\geq \frac{2v_e \Gamma(N)}{\Gamma(N+2/\alpha)}\sqrt{\frac{P_U}{\pi^3 P_D \eta}}$,
the lower bound given in Proposition \ref{proposition: rho_bound_closedForm} equals $1$, and thus we
obtain the following corollary.
\begin{corollary} \label{corollary: rho=1}
For the   finite-capacity battery case, if $\lambda_{AP}\geq \frac{2v_e \Gamma(N)}{\Gamma(N+2/\alpha)}\sqrt{\frac{P_U}{\pi^3 P_D \eta}}$, $\rho=1$.
\end{corollary}

Although the  lower and upper bounds provided in Proposition \ref{proposition: rho_bound_closedForm}
are in closed-form, their tightness to the actual $\rho$ of the finite-capacity
case cannot be assured for arbitrary $C$ and other parameters. Thus, in the
following, we provide an alternative   lower bound to $\rho$ which is tight in
general.

\subsubsection{Tight Lower Bound of Transmission Probability $\rho$}
\label{section: tight_lower_bound}
The tight lower bound of $\rho$ is obtained by modeling
the battery level as a discrete-time Markov chain \cite{Huang.IT.13}.
In the following, we  first
quantize  $C$, $Z_F$, and $P_U$, and then based on the resulting
battery level, we develop the  discrete-time Markov chain with finite number of
states. By finding the
steady-state probabilities of the Markov chain,  we
novelly derive a tight lower bound to the transmission probability $\rho$,
which is not in closed-form but can be computed efficiently.

First, we quantize the battery capacity $C$, the harvested energy $Z_F$,
and the required transmit power $P_U$ of the typical wireless node,
such that the battery level only has a finite number of values. Specifically,
let $\delta\ll C$
represent the quantization step size, which assures   $\lceil
P_U/\delta \rceil\leq \lfloor C/\delta \rfloor$, with $\lceil x
\rceil$ and $\lfloor x \rfloor$ denoting  ceiling and floor operations of
$x\in\mathbb{R}$, respectively. We reduce $C$ and
  $Z_F$ to $\delta \lfloor C/\delta \rfloor$  and   $ \delta
\lfloor Z_F/\delta \rfloor$, respectively, and increase   $P_U$ to $ \delta
\lceil P_U/\delta \rceil$. Clearly, under these operations, the resulting
battery level is a lower bound to  $S_F$ in (\ref{eq: S_F_Storage}),
which is denoted by $S_F^{LB}$,  given as \vspace{-0.05in}
\begin{align}
\!\!\!S_F^{LB}=&\min\Big(\delta \lfloor C/\delta \rfloor,\nonumber\\
&\!\!\!\!\!\!\!\!\!S_{F-1}^{LB} -\delta \lceil P_U/\delta
\rceil I\left(S_{F-1}^{LB}\!\geq\! \delta \lceil P_U/\delta
\rceil \right)\!+\!\delta \lfloor Z_F/\delta
\rfloor\Big)  \label{eq: S_F_Storage_low}
\end{align}
with  initial $S_0^{LB}=0$. For any $F\geq 0$, we have $S_F^{LB}\in \{0,
\delta,...,\delta\lfloor C/\delta \rfloor \}$. By replacing $S_F$ with
$S_F^{LB}$ in (\ref{eq: trans_prob_def}), we  obtain a lower bound to
$\rho$, which is denoted by $\rho^{LB}$. It is easy
to verify that when $\delta$ is
sufficiently small, $\rho^{LB}$ is a tight lower bound to $\rho$, which is
expected to outperform the bounds in Proposition \ref{proposition: rho_bound_closedForm}. Moreover,
when
$\delta\rightarrow 0$,  we have $\rho^{LB}=\rho$ due to $S_F^{LB}=S_F$.

Next, we derive $\rho^{LB}$ by
analyzing the distribution of $S_F^{LB}$ via Markov-chain theory. Let
$U=\lceil P_U/\delta \rceil$ and $V=\lfloor C/\delta \rfloor$. From
(\ref{eq: S_F_Storage_low}), given $S_{F-1}^{LB}$ with $F\geq 2$, $S_F^{LB}$
is independent
of $\{S_n^{LB}\}_{t=0}^{F-2}$. Thus,  $\{S_F^{LB}\}$  satisfies the Markov
property and is hence a discrete-time Markov chain, with the state space given
by $\{0, \delta,..., V \delta\}$.
Let $P_{ij}=\mathbb{P}\left(S_F^{LB}=j\delta \big| S_{F-1}^{LB}=i\delta
\right)$ represent the transition probability from state
$i\delta$ to $j \delta$, with $i, j\in \{0, ...,V\}$.
If $j<V$, the battery level $j\delta$ is below the capacity limit $V
\delta$, and thus \vspace{-0.05in}
\begin{align}
P_{ij}&=\mathbb{P}\big(\delta  \lfloor Z_F/\delta
\rfloor=(j-i)\delta+U \delta I(i\geq U)\big)  \nonumber \\
&=\mathbb{P}\big((j-i)\delta+U \delta I(i\geq U) \leq Z_F<
\nonumber \\
&~~~~~~~~~~~~~~~~~~~~~~~(j-i+1)\delta +U \delta I(i\geq U) \big).
\label{eq: P_ji_j}
\end{align}
If $j=V$, state transition from $i$ to $j$ includes all events that can cause
battery saturation, and thus \vspace{-0.05in}
\begin{align}
P_{ij}&=\sum_{k=V-i}^{\infty}\mathbb{P}\big(\delta  \lfloor Z_F/\delta
\rfloor=k \delta+U \delta I(i\geq U)\big)  \nonumber \\
&=\mathbb{P}\big( Z_F \geq
k \delta+U \delta I(i\geq U) \big). \label{eq: P_ji_v0}
\end{align}
By combining (\ref{eq: P_ji_j}) and (\ref{eq: P_ji_v0}), we obtain
\begin{equation}
 P_{ij} \!\!=\!\!\left\{\!\!\!
   \begin{array}{l}
     \mathbb{P}\big((j\!-\!i)\delta \!\leq \!Z_F\!<\! (j\!+\!1\!-\!i)\delta\big), \\
     ~~~~~~~~~~~~~~~~~~~~~~~~~~~~~~~~\textrm{if}~j\!<\!V, i\!<\!U,\\
     \mathbb{P}\big((j\!-\!i)\delta \!+\!U\delta\! \leq \!Z_F \!<\!
(j\!+\!1\!-\!i)\delta\!+\!U\delta\big),\\
    ~~~~~~~~~~~~~~~~~~~~~~~~~~~~~~~~\textrm{if}~j\!<\!V, i\!\geq \!U,\\
    \mathbb{P}\big(Z_F\!\geq \!(V\!-\!i)\delta\big),  ~~~~~~~~~~~\!\!\textrm{if}~j\!=\!V, i\!<\!U,\\
    \mathbb{P}\big(Z_F\!\geq\! (V\!-\!i)\delta\!+\!U\delta\big), ~~~~~\!\!\textrm{if}~j\!=\!V, i\!\geq\! U,
   \end{array} \label{eq: state_trans_Prob}
  \right.
\end{equation}
where in each case,  $P_{ij}$ is only determined by the distribution of $Z_F$.
Denote $\boldsymbol{\pi}=[\pi_0,...\pi_{V}]$ as the
steady-state probabilities of the Markov chain, and  $\boldsymbol{P}$ as
the state transition probability matrix with the $(i,j)$-th element given by
$P_{ij}$. By jointly solving $\boldsymbol{\pi P}=\boldsymbol{\pi}$ and
$\sum_{i=0}^{V}\pi_i=1$, or applying $\boldsymbol{\pi 'P}^k=\boldsymbol{\pi}$
with a randomly initialized state probabilities $\boldsymbol{\pi'}=[\pi_0',...\pi_{V}']$
and $k\in\mathbb{Z}$,
we can find the value of $\pi_i$, $\forall
i\in[0,...,V]$, and thus obtain $\rho^{LB}=\sum_{i=0}^{U}\pi_i$.
Since there is no general expression to each $\pi_i$, $\rho^{LB}$ can only be
obtained numerically in general.  In Algorithm 2, by reducing $\delta$ to
repeatedly calculate $\rho^{LB}$ until an absolute error bound, denoted by
$\theta\ll 1$ is satisfied, we present a simple procedure to calculate
$\rho^{LB}$, which ensures  $|\rho-\rho^{LB}|\leq \theta$.

\begin{algorithm}
\caption{Markov-chain based search algorithm to find a tight lower bound to
$\rho$}
\begin{algorithmic}[1]
\STATE initialize $\delta$ and $\theta$, and set $\rho_0=1$ and $\rho^{LB}=0$.
\WHILE{$|\rho_0-\rho^{LB}|>\theta$}
\STATE set $\rho_0=\rho^{LB}$.
\STATE set $\delta=\delta/2$.
\STATE calculate $U$, $V$, and $\boldsymbol{P}$.
\STATE find $\boldsymbol{\pi}$, such that $\boldsymbol{\pi}=\boldsymbol{\pi P}$,
and set $\rho^{LB}=\sum_{i=0}^{U}\pi_i$.
\ENDWHILE
\STATE return $\rho^{LB}$.
\end{algorithmic}
\end{algorithm}
%\vspace{-0.3in}
\begin{remark}
The computational complexity of Algorithm 2 is determined by the values of
$\delta$ and $\theta$ as well as the efficiency to find the steady-state
probabilities of $\boldsymbol{\pi}$. As a result, it is generally difficult to
find the complexity order of Algorithm 2 analytically, as   in
\cite{Weber.Alg}.
Intuitively,  when $\delta$ is very small,  due to the
resulting large size of the state transition probability matrix
$\boldsymbol{P}$, Algorithm 2
may not be computationally efficient. However, it is worth noting that Algorithm
2 is essentially an {\it off-line} algorithm. Moreover,
since it is not only difficult to find an exact expression of $\rho$,  but also
computationally
prohibitive  to obtain $\rho$ by network-level simulation,   the tight lower
bound $\rho^{LB}$ provided by Algorithm 2 is   important for analytical
 study of  the actual transmission probability and thus the spatial
throughput. For example, as will be shown later,
$\rho^{LB}$ can help  evaluate
 the performance of other lower or upper bounds, and maximize the spatial
throughput for a finite-capacity battery
case with any designed parameters.
\end{remark}

\vspace{0.02in}
\subsubsection{Spatial Throughput Maximization} \label{section: STM_finite}
We consider two cases, which are $\lambda_{AP}\!\geq\! \frac{2v_e \Gamma(N)}{\Gamma(N\!+\!2/\alpha)}\sqrt{\frac{P_U}{\pi^3 P_D \eta}}$
and $\lambda_{AP}\!< \!\frac{2v_e \Gamma(N)}{\Gamma(N\!+\!2/\alpha)}\sqrt{\frac{P_U}{\pi^3 P_D \eta}}$,
respectively, for spatial throughput
maximization. We have $\rho=1$ from Corollary \ref{corollary: rho=1} in
the former case, but no exact expression of $\rho$ in the latter case.

First, we consider the case with $\lambda_{AP}\!\geq \!\frac{2v_e \Gamma(N)}{\Gamma(N+2/\alpha)}\sqrt{\frac{P_U}{\pi^3 P_D \eta}}$.
Similar to the infinite-capacity battery case, when $\rho\!=\!1$, the spatial
throughput
$R(N,P_U)$ becomes a constant. Thus, by substituting $\rho\!=\!1$ into Problem
(P2), and adding the constraint
$\lambda_{AP}\!\geq \!\frac{2v_e \Gamma(N)}{\Gamma(N+2/\alpha)}\sqrt{\frac{P_U}{\pi^3 P_D \eta}}$, or
equivalently,
$P_U\!\leq \!\pi^3 P_D \eta \left[\frac{ \lambda_{AP} \Gamma(N+2/\alpha)}{2 \Gamma(N) v_e} \right]^2$,  the
spatial throughput maximization problem degenerates to a  feasibility
problem, given by
\begin{align}
%\vspace{-0.5in}
\textrm{(P5)}~~~\mathrm{Find}&~N, P_U
\nonumber \\
\mathrm{such~ that}& ~N\!\in \!\left\{\!1,..., \min\!\left(\!T\!-\!1,
T\!-\!\frac{\lambda_w}{K_{\epsilon}\lambda_{AP}}\! \right)\!\right\},  \label{eq:
N_finite} \\
& ~P_{\min}\leq P_U \leq \min
\Big(\!P_{\max},  \nonumber \\
&~~~~~~~~~~~~~~\pi^3 P_D \eta \left[\frac{ \lambda_{AP} \Gamma(N\!+\!2/\alpha)}{2 \Gamma(N) v_e} \right]^2 \!\Big). \label{eq:
P_finite}
\end{align}

It is observed that the optimal solutions $N^{*}$ and $P_U^{*}$  to Problem (P5)
are arbitrary values that locate in the feasible region,
defined by (\ref{eq: N_finite}) and (\ref{eq: P_finite}).
It is also observed that due to the reduced battery capacity, the feasible
region of Problem (P5) for the finite-capacity battery case is reduced, as
compared to that of Problem (P4) for the infinite-capacity battery case.
Unlike Problem (P4), where $N^{*}$ and $P_U^{*}$ can be independently selected
in its feasible region, $N^{*}$ and $P_U^{*}$ of Problem (P5) may be correlated,
due to the added constraint $P_U\!\leq \! \pi^3 P_D \eta \left[\frac{ \lambda_{AP} \Gamma(N+2/\alpha)}{2 \Gamma(N) v_e} \right]^2$
to ensure $\rho \!= \!1$. Similar to
both  battery-free and infinite-capacity battery cases, due to the non-decreased feasible region, we also  find that
the number of optimal solutions is non-decreasing over the AP density $\lambda_{AP}$.

Next, we focus on the case with $\lambda_{AP}\!< \!\frac{2v_e \Gamma(N)}{\Gamma(N+2/\alpha)}\sqrt{\frac{P_U}{\pi^3 P_D \eta}}$. Due to the lack
of exact expression of $\rho$ and thus $R(N,P_U)$, given in  (\ref{eq:
spatial_throughput_def}),  we exploit
Algorithm 2 to study the spatial throughput maximization problem, as defined by
Problem (P2). Specifically, for any $P_U\!\in\![P_{\min}, P_{\max}]$ and
$N\!\in\!\{1,...,T\!-\!1\}$, we first apply Algorithm 2 to find a tight lower bound
$\rho^{LB}$ to $\rho$. Then based on $\rho^{LB}$, if the transmission
probability constraint
$\lambda_w\rho^{LB}\!\leq\! K_{\epsilon} \lambda_{AP} (T\!-\!N)$   is
satisfied,   we can obtain a non-zero tight lower bound of  the spatial
throughput
$R(N,P_U)$, which is denoted by $R^{LB}(N,P_U)$; otherwise, we set
$R^{LB}(N,P_U)\!=\!0$. After finding all $R^{LB}(N,P_U)$'s over $P_U\!\in\![P_{\min},
P_{\max}]$ and $N\!\in\!\{1,...,T\!-\!1\}$, we can easily find the  optimal solutions
$N^{*}$ and $P_U^{*}$ that maximizes $R^{LB}(N^{*},P_U^{*})$.
%Note that since $R^{LB}(N,P_U)$ is determined by $\rho^{LB}$
From  (\ref{eq: spatial_throughput_def}), if $\lim_{\delta\rightarrow 0} \rho^{LB}\!=\!\rho$ by
adopting a sufficiently small $\theta$ in Algorithm 2,
we  have  $\lim_{\delta\rightarrow 0} R^{LB}(N,P_U)\!=\!R(N,P_U)$, over
any $P_U\!\in\![P_{\min}, P_{\max}]$ and $N\!\in\!\{1,...,T\!-\!1\}$.
Therefore, the obtained
$N^{*}$ and $P_U^{*}$ can be seen as tight approximations to
the actual optimal DL slots and UL transmit power, respectively.
A numerical example is provided in Section VI-B to find the maximized spatial
throughput based on Algorithm~2.
%\vspace{-0.01in}

\section{Numerical Results}
Numerical results are provided in this section. In the following, we first  validate the analytical results, and then further study  the transmission probability and spatial throughput for both battery-free and battery-deployed cases.

\subsection{Validation of the Analytical Results}
This subsection  validates  the analytical results obtained in Section II and Section III by simulation. We  validate the feasibility of Assumption 1 for independent $Z_F$'s, and the homogeneous PPP assumption for  the point process formed by the active wireless nodes in the UL slot. We also find that the distribution of   $Z_F$  in Proposition \ref{proposition: Z_F_distribution} and    $P_{suc}$  in Proposition \ref{proposition: P_suc}  can be similarly validated by using the methods in the existing literature (e.g., \cite{Haenggi.book} and \cite{Andrews.COM.11}).    We focus on the battery-free case in Type-I network model, and find similar validation  results   for  the battery-deployed case in Type-I network model as well as  both cases in  Type-II network model.
Specifically, at the beginning of each frame, we generate $\Phi(\lambda_{AP})$ for APs and $\Phi(\lambda_w)$ for wireless nodes  in  a square of
$[0\textrm{m},1000\textrm{m}]\!\times\![0\textrm{m},1000\textrm{m}]$, according to the method described in \cite{Stoyan.SG.95}. Then at the beginning of each slot within a frame, we  independently and uniformly relocate all the wireless nodes in the considered area. To take care of the border effects, we focus on sampling the wireless nodes that locate in the interim square with side length $L$m, $0<L<1000$.  Unless otherwise specified, in this subsection, we set $\eta=0.4$, $\lambda_w=0.005/m^2$, $P_D=10$W, $T=3$, and $N=2$. All  simulation results   are obtained based on an average over 4000 frame realizations.

\begin{figure}
\setlength{\abovecaptionskip}{-0.1in}
\centering
\DeclareGraphicsExtensions{.eps,.mps,.pdf,.jpg,.png}
\DeclareGraphicsRule{*}{eps}{*}{}
\includegraphics[angle=0, width=0.47\textwidth]{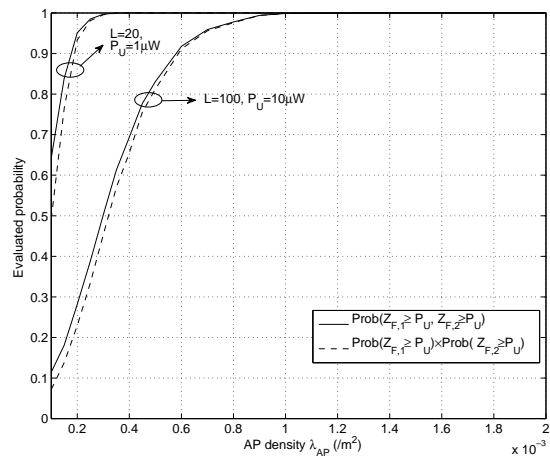}
\caption{Validation of the independent assumption for $Z_F$'s.}
\label{fig: Z_F}
\vspace{-0.15in}
\end{figure}
First, we validate the feasibility of Assumption 1.  Since  the correlations between $Z_F$'s are similar over time frames and space $\mathbb{R}^2$, we focus on validating that $Z_F$'s can be viewed independent over space $\mathbb{R}^2$.
Specifically, we randomly select two wireless nodes and index them with $i=1$ and $i=2$, respectively. The two wireless nodes  independently and uniformly change their locations over frames in the interim square with length $L$m. We consider two scenarios, where in the first scenario, we set $L=20$m  and   $P_U=1\mu$W, and in the second scenario, we set $L=100$m and $P_U=10\mu$W. Clearly, both wireless nodes are of smaller mobility in the former scenario and larger mobility in the latter one. Moreover, since $L<\infty$,  the two wireless nodes are of limited mobility over frames in both scenarios.
In Fig.~\ref{fig: Z_F}, we evaluate and compare the marginal probability product $\mathbb{P}(Z_{F,1}\geq P_U) \times \mathbb{P}(Z_{F,2}\geq P_U)$ with the joint probability $\mathbb{P}(Z_{F,1}\geq P_U,Z_{F,2}\geq P_U)$ over the AP density $\lambda_{AP}$ in both scenarios.
It is observed from  Fig.~\ref{fig: Z_F} that in both scenarios, for any AP density, the gap between the  marginal probability product and the joint probability is tightly approaching  zero; and especially when AP density is reasonably large, such gap decreases to be zero. Hence,  the  harvested energy of these two wireless nodes in one frame is tightly approaching to be independent, and thus can be viewed as independent.  Moreover, by comparing the two scenarios, it is observed that when the wireless nodes are of smaller mobility in the first scenario, the gap between  marginal probability product and the joint probability is comparatively large when $\lambda_{AP}$ is quite small  (e.g., $\lambda_{AP}=0.0001/\textrm{m}^2$). This is mainly because when $\lambda_{AP}$ is quite small, the dominate APs are more correlated for wireless nodes with smaller mobility, as compared to that with larger mobility. However, the resulted correlation  is  rapidly reduced as $\lambda_{AP}$ is
reasonably increased.  Therefore, Assumption 1 can be well applied in the considered WPCN.

Next, we validate the Poisson assumption for the point process formed by the active wireless nodes in the UL slot. According to \cite{Stoyan.SG.95}, a  point process on $\mathbb{R}^2$ is \emph{fully} characterized by its void probability on an arbitrary compact subset of $\mathbb{R}^2$. We  evaluate and compare  the void probability of the  actual point process in the UL slot with that of the assumed PPP $\Phi(\lambda_a)$ in the interim square with side length $L$,  by setting   $L=1:1:20$, in  Fig. \ref{fig: void}.    From  \cite{Stoyan.SG.95}, given $L$, the void probability of $\Phi(\lambda_a)$ in the  interim square   is given by $\exp\left(-\lambda_a L^2\right)$.  We set $\lambda_{AP}=0.0005/\textrm{m}^2$ and $P_U=10\mu$W.
 It is observed  from Fig. \ref{fig: void} that the void probabilities of both the assumed PPP and the actual point process in the UL decrease  over the increased interim area with side length $L$, as expected.  Moreover,  since Assumption 1 can be well applied, as its direct result to obtain the PPP $\Phi(\lambda_{a})$ in the UL, it is observed that for any $L$, the void probability of the assumed PPP $\Phi(\lambda_{a})$  is tightly close to that of the actual point process in the UL,  which validates the PPP assumption for the point process in the UL. In addition, from (\ref{eq: trans_prob_def}) and (\ref{eq: lambda_a}), since the density $\lambda_a$  is determined by the distribution of $Z_F$, the successful validation of the assumed   PPP $\Phi(\lambda_a)$    also implies the correctness of the derived $Z_F$'s distribution in Proposition 2.1 under Assumption~1.

\begin{figure}
\setlength{\abovecaptionskip}{-0.1in}
\centering
\DeclareGraphicsExtensions{.eps,.mps,.pdf,.jpg,.png}
\DeclareGraphicsRule{*}{eps}{*}{}
\includegraphics[angle=0, width=0.47\textwidth]{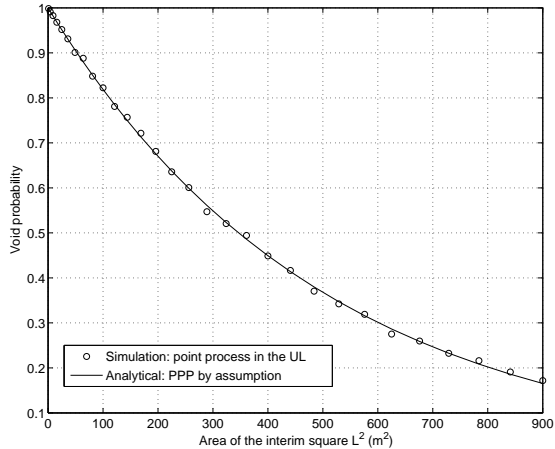}
\caption{Validation of the PPP assumption  for the point process in the UL.}
\label{fig: void}
\vspace{-0.147in}
\end{figure}

\begin{figure}
\setlength{\abovecaptionskip}{-0.1in}
\centering
\DeclareGraphicsExtensions{.eps,.mps,.pdf,.jpg,.png}
\DeclareGraphicsRule{*}{eps}{*}{}
\includegraphics[angle=0, width=0.47\textwidth]{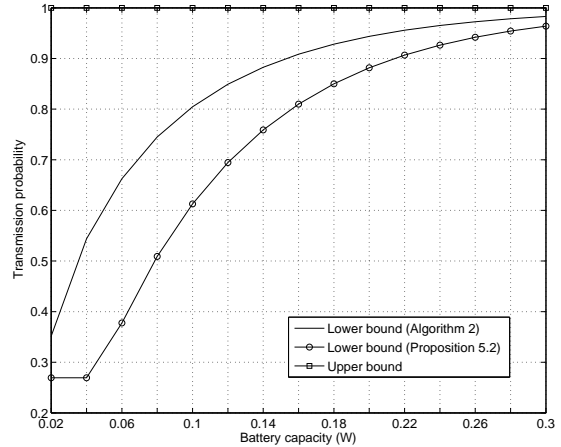}
\caption{Transmission probability over battery capacity.}
\label{fig: batter_finite_rho}
\vspace{-0.155in}
\end{figure}

\subsection{Study on Transmission Probability  and Spatial Throughput}
This subsections studies the transmission probability and the spatial throughput.
Unless otherwise specified, in this subsection, we reasonably set $P_D=10$W,
$\sigma^2=-60$dBm, $\epsilon=0.05$, $\eta=0.4$, $T=100$, and $\beta=5$. Moreover,
we set $n=9$ for calculating $\textrm{erf}(x)$ in (\ref{eq: erf}), where $v_e$
is obtained as $v_e=\textrm{erfinv}(10^{-9})$. Similarly, $g_0$ is  obtained by
numerically solving    $g_0 Q\left(\frac{g_0}{2 \pi}\right) =(1-\epsilon)\exp
\left(-\frac{g_0^2}{4 \pi}\right)$ as given in Proposition \ref{proposition:
P_suc_constraint_Equivalent}.
We also observe  that similar performance can be obtained by using other parameters.
\subsubsection{Transmission Probability $\rho$}
Since the transmission probabilities for the battery-free and infinite-capacity
battery cases
are obtained exactly, as given in (\ref{eq: rho_No_storage}) and (\ref{eq:
rho_storage_Inf}), respectively,
we focus on the transmission probability for the
finite-capacity battery case. We set $\lambda_w=0.0012/m^2$,
$\lambda_{AP}=0.0008/m^2$, $N=60$, and $P_U=0.02$W.

Fig.~\ref{fig: batter_finite_rho}  compares   the closed-form  lower and upper
bounds of $\rho$, given in Proposition \ref{proposition: rho_bound_closedForm},
and the tight lower bound, given by Algorithm 2, over the battery capacity $C$.
By adopting Algorithm 2, we set the absolute error $\theta=0.001$, and
initialize $\delta=0.0001$.
First, it is observed from Fig.~\ref{fig: batter_finite_rho} that the tight
lower bound by Algorithm~2 monotonically increases over battery capacity $C$ as
expected; and as the actual transmission probability, it is  bounded by the
upper and lower bounds provided in Proposition \ref{proposition:
rho_bound_closedForm}, respectively.
Next, for the closed-form lower bound by Proposition
\ref{proposition: rho_bound_closedForm},  it is observed
when the capacity is small with $0.2\leq C \leq 0.4$,   a constant lower
bound is obtained as $\textrm{erf}\left(\frac{\Gamma(N+2/\alpha) \lambda_{AP}}{2\Gamma(N)}\sqrt{\frac{\pi^3 P_D\eta}{P_U}} \right) \geq 1-e^{-Q(C-P_U)}$; and when $C>0.4$,
the lower bound is given by $ 1-e^{-Q(C-P_U)} \geq \textrm{erf}\left(\frac{\Gamma(N+2/\alpha) \lambda_{AP}}{2\Gamma(N)}\sqrt{\frac{\pi^3 P_D\eta}{P_U}} \right)$, which generally
captures the variation of the transmission probability,
by taking the tight lower bound by Algorithm 2 as a reference.
Moreover, as $C$ increases, we observe both lower bounds by Algorithm
2 and Proposition \ref{proposition: rho_bound_closedForm} approach to the upper
bound $\rho=1$, and that by Algorithm 2 becomes tight to $\rho=1$ when $C$ is
large.
Furthermore, noticing that $\textrm{erf}\left(\frac{\Gamma(N+2/\alpha) \lambda_{AP}}{2\Gamma(N)}\sqrt{\frac{\pi^3 P_D\eta}{P_U}} \right)$ is the transmission probability in the
battery-free case, given in (\ref{eq: rho_No_storage}), we observe  that it is
always lower than the tight lower bound by Algorithm 2 in the battery-deployed
case as expected.

\subsubsection{Spatial Throughput}
We study the spatial throughput  in both
 battery-free and battery-deployed cases.
In the battery-free case, by applying Theorem \ref{theorem: P-NB}, we focus on
showing the effects of the AP density $\lambda_{AP}$ and wireless node density
$\lambda_w$ on the maximized spatial throughput. In the battery-deployed case,
we focus on the challenging finite-capacity battery case with  $\lambda_{AP}< \frac{2v_e \Gamma(N)}{\Gamma(N+2/\alpha)}\sqrt{\frac{P_U}{\pi^3 P_D \eta}}$, and exploit Algorithm 2  to help find
the maximized spatial throughput.

Fig.~\ref{fig: batter_free_optimal_R} shows the maximized spatial throughput
over the AP density in the battery-free case, by applying Theorem  \ref{theorem:
P-NB}. We consider two scenarios, with wireless node density $\lambda_w= 0.0012/m^2$ and $\lambda_w=
0.002/m^2$, respectively, where for each scenario, the   low, medium and high AP
regimes are given by $\big[0, \lambda_w/(K_{\epsilon}(T-1)\big)$,
$\big[\lambda_w/(K_{\epsilon}(T-1)), \lambda_w/K_{\epsilon}\big)$, and
$\big[\lambda_w/K_{\epsilon}, \infty \big)$, respectively.
  First, for both scenarios,
it is observed  that  by
increasing $\lambda_{AP}$, the maximized spatial throughput slowly increases  in
the low AP density regime, and after
$\lambda_{AP}=\lambda_w/(K_{\epsilon}(T-1))$, it rapidly increases in the medium
AP density regime and achieves its maximum achievable spatial throughput
$\lambda_w \log_2(1+\beta)$ at some point in this regime; and after this point,
it remains as the constant $\lambda_w \log_2(1+\beta)$ over all the medium and
high AP density regimes. Since in   both scenarios, we observe that
  $\lambda_w \log_2(1+\beta)$ is achieved far before
$\lambda_{AP}$ reaches to its high density regime, for ease of illustration, we only
show the low AP density regime and part of the medium AP density regime in Fig.~ \ref{fig: batter_free_optimal_R}
for both scenarios.
  Next,   it is observed that the maximum achievable spatial
throughput $\lambda_w \log_2(1+\beta)$ is larger for the scenario with a
larger $\lambda_w= 0.002/m^2$, as compared to the scenario with $\lambda_w= 0.0012/m^2$.
Moreover, in the scenario with a larger $\lambda_w=
0.002/m^2$,  due to the increased interference level, to achieve
$\lambda_w \log_2(1+\beta)$ under the successful information transmission
probability constraint,   more APs are needed to be deployed  to  reduce the
distance between the wireless nodes and their associated APs, so as to improve the
desired signal strength  and thus the successful information transmission
probability.
\begin{figure}
\setlength{\abovecaptionskip}{-0.1in}
\centering
\DeclareGraphicsExtensions{.eps,.mps,.pdf,.jpg,.png}
\DeclareGraphicsRule{*}{eps}{*}{}
\includegraphics[angle=0,
width=0.47\textwidth]{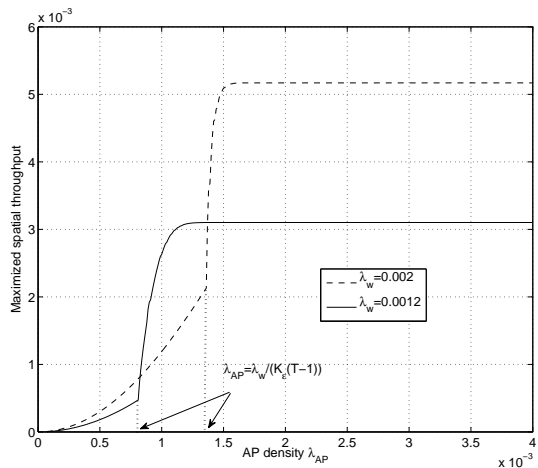}
\caption{Maximized spatial throughput over $\lambda_{AP}$ in battery-free case.}
\label{fig: batter_free_optimal_R}
\vspace{-0.15in}
\end{figure}

\begin{figure}
\setlength{\abovecaptionskip}{-0.1in}
\centering
\DeclareGraphicsExtensions{.eps,.mps,.pdf,.jpg,.png}
\DeclareGraphicsRule{*}{eps}{*}{}
\includegraphics[angle=0, width=0.47\textwidth]{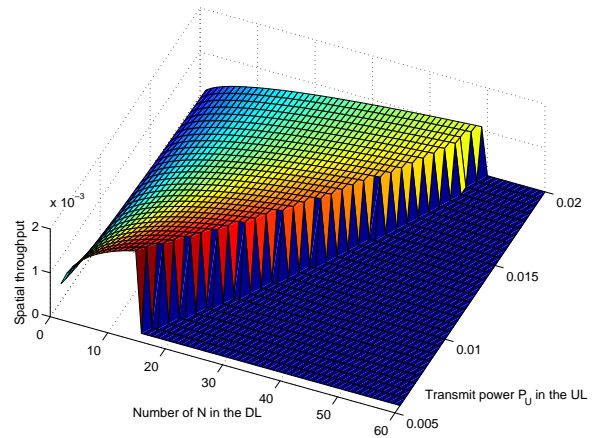}
\caption{Spatial throughput over feasible region in finite-capacity battery
case.}
\label{fig: batter_finite_R}
\vspace{-0.12in}
\end{figure}
Fig.~\ref{fig: batter_finite_R} shows the spatial throughput for the
finite-capacity battery case over  $N$ and
$P_U$, where we set $T=100$, $P_{\max}=0.02$W, $C=0.04$W, $\lambda_w=0.0012/m^2$, and
$\lambda_{AP}=0.0008/m^2$.
By applying Algorithm 2 with $\theta=0.001$ and initialized $\delta=0.0001$, we
use the method presented in Section \ref{section: STM_finite}
to compute   $R^{LB}(N,P_U)$ over all feasible $N$ and $P_U$, and take
the obtained $R^{LB}(N,P_U)$ as a tight approximation to $R(N,P_U)$. We
find  the optimal solutions that maximize  $R^{LB}(N,P_U)$ are
 $N^{*}=14$ and $P_U^{*}=P_{min}=0.0055$W in Fig.~\ref{fig: batter_finite_R}.
 Thus, similar to the battery-free case in Theorem \ref{theorem: P-NB}, the wireless nodes
 prefer to choose   $P_{min}$, which  assures a large transmission probability $\rho$.
 Moreover, with small $N^{*}=14$ in the DL phase, the UL phase is assigned with $T-N^{*}=86$ slots,
 which  helps effectively reduce the UL interference by the independent scheduling.
In addition,  since a  smaller $P_U$ yields an increased   $\rho$, and thus requires
a smaller $N$ to satisfy the transmission probability constraint in Problem (P2),
it is  observed  from Fig.~\ref{fig: batter_finite_R} that the feasible region of $N$ becomes smaller as $P_U$ decreases.

\section{Conclusion}
In this paper, we studied the optimal tradeoff between the energy
transfer and information transfer  in a large-scale WPCN, for both battery-free and battery-deployed wireless nodes.
We proposed a new time-partition-based harvest-then-transmit protocol and
modeled the network based on homogeneous PPPs.
By using tools from stochastic geometry, we characterized the distribution of
the harvested energy in the DL and the successful information transmission
probability in the UL.
We studied the resulting transmission probability and successfully solved the spatial
throughput maximization problem for both battery-free and battery-deployed cases.
Moreover, by comparing the network performance in the
battery-free,  infinite-capacity battery, and  finite-capacity
battery cases, we investigated the effects of battery storage on the system spatial throughput.

%%%%%%%%%%%%%%%%%%%%%%%%%%%%%%%%%%%%%%%%%%%%%%%%%%%%%%%%%%%%%%%%%%%%%%%%%%%%%%%%%%%%%%%%%%

\appendices

\section{Proof of Proposition \ref{proposition: P_suc_constraint_Equivalent}}
\label{appendix: proof_euqivalent_trans_prob}
We first  present three lemmas.
\begin{figure}
\centering
\DeclareGraphicsExtensions{.eps,.mps,.pdf,.jpg,.png}
\DeclareGraphicsRule{*}{eps}{*}{}
\includegraphics[angle=0, width=0.47\textwidth]{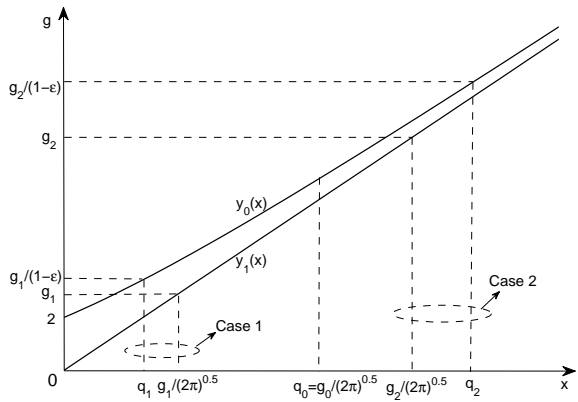}
\caption{Illustration of $y_0(x)=\frac{g}{1-\epsilon}$ and $y_1(x)=g$.}
\label{fig: equivalent_P_suc}
\vspace{-0.15in}
\end{figure}
\begin{lemma}\label{lemma: appendix_equivalent_P_suc_1}
For any $x\geq \frac{g}{\sqrt{2\pi}}$, $g\geq 0$, $g\exp\left(
\frac{x^2}{2}\right)Q(x)\geq 1-\epsilon$ is equivalent to
$x\leq q$ with $\frac{g}{\sqrt{2\pi}}\leq q$, where $q$ is the unique solution
to $g\exp\left( \frac{q^2}{2}\right)Q(q)= 1-\epsilon$.
\end{lemma}
\begin{proof}
Let $y_0(x)=\frac{\exp\left( \frac{-x^2}{2}\right)}{Q(x)}$. It is easy to verify
that $y_0(x)$ monotonically increases over $x\geq 0$, as shown in Fig.~\ref{fig:
equivalent_P_suc}. As a result, $g\exp\left( \frac{x^2}{2}\right)Q(x)\geq
1-\epsilon$  is equivalent to $x\leq q$. Moreover, since $x\geq
\frac{g}{\sqrt{2\pi}}$, we  need $\frac{g}{\sqrt{2\pi}}\leq q$; otherwise,
no valid $x\in[\frac{g}{\sqrt{2\pi}},q]$ exists to meet $g\exp\left(
\frac{x^2}{2}\right)Q(x)\geq 1-\epsilon$. Lemma \ref{lemma: appendix_equivalent_P_suc_1} thus
follows.
\end{proof}
\begin{lemma} \label{lemma: appendix_equivalent_P_suc_2}
$\frac{g}{\sqrt{2\pi}}\leq q$ is equivalent to $g\geq g_0$.
\end{lemma}
\begin{proof}
Let $y_1(x)\!=\!\sqrt{2 \pi}x$. While $q$ is the unique solution to
$y_0(x)\!=\!\frac{g}{1-\epsilon}$, we have
$\frac{g}{\sqrt{2 \pi}}$ is the unique solution to $y_1(x)\!=\!g$. Notice that $q$
is a function of $g$, and as   shown in Fig.~\ref{fig: equivalent_P_suc}, when
$g$ increases, both $q$ and $\frac{g}{\sqrt{2 \pi}}$ increase.
First, since $g_0$ is the unique
solution to $\frac{g_0}{1-\epsilon}\!=\!\frac{\exp \left(-g_0^2/4 \pi
\right)}{Q\left(g_0/2 \pi \right)}$,
it is easy to obtain that when $g\!=\!g_0$,
$q_0\!=\!\frac{g_0}{\sqrt{2 \pi}}$ is the unique solution to
$\frac{g_0}{1-\epsilon}\!=\!\frac{\exp\left( -q_0^2/2\right)}{Q(q_0)}$.
In other words, when $g\!=\!g_0$, we obtain $q\!=\!\frac{g}{\sqrt{2 \pi}}$. Next, by expanding
$Q(x)$, we have   $Q(x)\!=\!\frac{1}{\sqrt{2 \pi}}\exp\left(
-\frac{x^2}{2}\right)\left[
\frac{1}{x}-\frac{1}{x^2}+o(x^{-4})\right]$. We thus can easily obtain
$y_0(x)\geq y_1(x)$, and $\lim_{x\rightarrow \infty}y_0(x)\!=\!y_1(x)$, i.e.,
$y_0(x)$ and $y_1(x)$ are getting closer as $x$ increases. As a result, as
illustrated in Fig.~\ref{fig: equivalent_P_suc}, it is
easy to verify the followings:
 1) when $g<g_0$, due to the big gap between $y_0(x)$ and $y_1(x)$,   we have
$q<\frac{g}{\sqrt{2 \pi}}$, as illustrated by Case 1 with $g=g_1$ and $q=q_1$;
and 2) due to the increasingly small gap between $y_0(x)$ and $y_1(x)$ as $g$
increases, when $g>g_0$, we have    $q>\frac{g}{\sqrt{2 \pi}}$, as illustrated
by Case 2 with $g=g_2$ and $q=q_2$. Lemma \ref{lemma:
appendix_equivalent_P_suc_2} thus follows.
\end{proof}
\begin{lemma}\label{lemma: appendix_equivalent_P_suc_3}
When $g\geq g_0$, we have
$\frac{\exp\left(-\frac{x^2}{2}\right)}{Q(x)}=\sqrt{2\pi}x$, as
$\epsilon\rightarrow 0$.
\end{lemma}
\begin{proof}
It has been shown from the proof of Lemma \ref{lemma:
appendix_equivalent_P_suc_2} that
$\frac{\exp\left(-\frac{x^2}{2}\right)}{Q(x)}\!>\!\sqrt{2\pi}x$. On the other hand,
when $g\!\geq \!g_0$, since  $q\!\geq \!\frac{g}{\sqrt{2 \pi}}$, we have
$\frac{\exp\left(-\frac{x^2}{2}\right)}{Q(x)}\Big|_{x=\frac{g}{\sqrt{2
\pi}}}\!\leq\!
\frac{\exp\left(-\frac{x^2}{2}\right)}{Q(x)}\Big|_{x=q}\!=\!\frac{g}{1-\epsilon}$,
or equivalently, $\frac{\exp\left(-\frac{x^2}{2}\right)}{Q(x)} \leq
\frac{\sqrt{2\pi}x}{1-\epsilon}$, by substituting $g\!=\!\sqrt{2\pi}x$.
As a result, $\frac{\exp\left(-\frac{x^2}{2}\right)}{Q(x)}\!=\!\sqrt{2\pi}x$, as
$\epsilon\rightarrow 0$. Lemma \ref{lemma: appendix_equivalent_P_suc_3} thus
follows.
\end{proof}

Therefore, since $\Upsilon\geq \frac{G}{\sqrt{2\pi}}$ in (\ref{eq: P_suc_4}),
from Lemma \ref{lemma: appendix_equivalent_P_suc_1} and Lemma  \ref{lemma:
appendix_equivalent_P_suc_2}, (\ref{eq: P_suc_4}) is equivalent to $\Upsilon \leq q$ with $G\geq g_0$.
From Lemma  \ref{lemma: appendix_equivalent_P_suc_3},
 we can obtain $q=\frac{G}{(1-\epsilon)\sqrt{2\pi}}$ with
$\epsilon\rightarrow 0$. By
substituting the expression of $\Upsilon$ and $G$, we find $P_{suc}\geq
1-\epsilon$ is equivalent to the transmission probability constraint with
$P_U\geq \frac{g_0^2\beta \sigma^2}{\pi^3 \lambda_{AP}^2}$. Proposition
\ref{proposition: P_suc_constraint_Equivalent} thus follows.

\section{Proof of Theorem \ref{theorem: P-NB}} \label{appendix: proof_theorem}
Note that in the first constraint of (P3), for the left-hand side, we have $
\lambda_w\textrm{erf}\left(\frac{\Gamma(N+2/\alpha) \lambda_{AP}}{2\Gamma(N)}\sqrt{\frac{\pi^3 P_D\eta}{P_U}} \right))\leq \lambda_w$,
 and for the right-hand side, we have
$K_{\epsilon}\lambda_{AP} \leq K_{\epsilon}\lambda_{AP} (T-N)\leq
K_{\epsilon}\lambda_{AP} (T-1)$. Thus, by comparing the upper bound of $
\lambda_w\textrm{erf}\left(\frac{\Gamma(N+2/\alpha) \lambda_{AP}}{2\Gamma(N)}\sqrt{\frac{\pi^3 P_D\eta}{P_U}} \right)$ with the upper and lower bounds of
$K_{\epsilon}\lambda_{AP} (T-N)$, respectively, we obtain the following three
regimes of the AP density $\lambda_{AP}$:
\begin{enumerate}
\item If $\lambda_w\leq K_{\epsilon}\lambda_{AP}$,
i.e., in the high AP density regime, it is clear the first
constraint in Problem (P3) can
always hold. Thus, any $N\in\{1,...,T-1\}$
and $P_U\in[P_{\min},P_{\max}]$ are feasible to Problem (P3).  Note that
$R(N,P_U)$ achieves its maximum value when $N=T-1$ and $P_U=P_{\min}$. As a
result, if $N=T-1$ and $P_U=P_{\min}$ satisfy (\ref{eq: rho_1_v_e}) for assuring
$\rho=1$, we find any pair of $N^{*}$
and $P_U^{*}$ that satisfy (\ref{eq: rho_1_v_e})
 are optimal to Problem (P3), and $R(N^{*},P_U^{*})=\lambda_w\log(1+\beta)$;
otherwise, we have $\rho<1$ and thus $R(N^{*},P_U^{*})<\lambda_w\log(1+\beta)$,
with $N^{*}=T-1$ and $P_U^{*}=P_{\min}$.
 \item If $K_{\epsilon}\lambda_{AP}<\lambda_w\leq
K_{\epsilon}\lambda_{AP}(T-1)$, i.e., in the medium AP density regime,  a
unique $N_0$ clearly exists, since otherwise, the condition
$K_{\epsilon}\lambda_{AP}<\lambda_w\leq K_{\epsilon}\lambda_{AP}(T-1)$ cannot
hold. It is thus easy to verify that the first constraint in Problem (P3)
holds if and only if $N\leq
N_0$. As a result, the feasible region for Problem (P3) is given by any
$N\in\{1,...,N_0\}$ and $P_U\in[P_{\min},P_{\max}]$. At last, by using the
similar method  as in the case of  high AP density regime, we can easily find
$N^{*}$, $P_U^{*}$ and $R(N^{*},P_U^{*})$ as stated in Theorem \ref{theorem:
P-NB}.
\item If $\lambda_w> K_{\epsilon}\lambda_{AP}(T-1)$, i.e., in the low AP density
regime, if $ \lambda_w\textrm{erf}\left(\frac{\Gamma(N+2/\alpha) \lambda_{AP}}{2\Gamma(N)}\sqrt{\frac{\pi^3 P_D\eta}{P_{max}}} \right) > K_{\epsilon}\lambda_{AP} (T-N)$
at $N=1$, which gives the largest value of the right-hand side in the first constraint of (P3),  the first constraint of Problem (P3) cannot hold,
and thus there is no feasible solution; otherwise, there exists
optimal $N^{*}$ and $P_U^{*}$, which yield $\rho<1$.
As shown in Algorithm 1, since for any given $N$,
 $\lambda_w\textrm{erf}\left(\frac{\Gamma(N+2/\alpha) \lambda_{AP}}{2\Gamma(N)}\sqrt{\frac{\pi^3 P_D\eta}{P_U}} \right)$ achieves its minimum value when
$P_U=P_{\max}$, we use $\lambda_w \textrm{erf}\left(\frac{\pi^2 \lambda_{AP}
N}{4}
\sqrt{\frac{P_D\eta}{P_{\max}}} \right) \leq K_{\epsilon}\lambda_{AP} (T-N)$ to
check whether an $N$ is feasible, by searching over $N\in\{1,...,T-1\}$. After
finding a
feasible $N$, we then calculate the corresponding $P_U=\max(P_{s},P_{\min})$
that
maximizes $R(N, P_U)$. Finally, by comparing all the feasible $N$'s and their
corresponding $P_U$'s, we can find optimal $N^{*}$ and $P_U^{*}$ that maximizes
$R(N,P_U)$. Clearly, by  searching over
$N\in\{1,...,T-1\}$, Algorithm 1 is of complexity or of $\mathcal{O}(T)$.
\end{enumerate}
Based on the above three cases, Theorem \ref{theorem: P-NB} thus follows.

\section{Proof of Proposition \ref{proposition: rho_storage_Inf}}
\label{appendix: proof_rho_inf_B}
We note  a  different proof based on random walk  theory for Proposition
\ref{proposition: rho_storage_Inf}  was provided in \cite{Huang.IT.13}.
Compared to \cite{Huang.IT.13}, by exploiting the distribution of
$Z_F$, the proof presented in the following is much simpler.
From (\ref{eq: S_F_Storage}), we have
\begin{align}
S_F=\sum_{i=1}^{F}Z_i-P_U\sum_{i=1}^{F}I(S_{i-1}\geq P_U)
\geq \sum_{i=1}^{F}Z_i-FP_U. \label{eq: appx_S_F_geq}
\end{align}
Under Assumption 1 with i.i.d. $Z_F$'s, it is easy to verify that the point processes at the end of the DL phase
of each frame are i.i.d. PPPs, each with density $\lambda_{AP}$. Thus, $\sum_{i=1}^{F}Z_i$ gives  the harvested energy over all $F$ i.i.d PPPs,
which is equivalent to the harvested energy in a PPP of density $F\lambda_{AP}$.
Hence,  we can easily obtain that for any given $z\geq0$,
$\mathbb{P}\left(\sum_{i=1}^{F}Z_i \geq z \right)=\textrm{erf}\left(\frac{\Gamma(N+2/\alpha) F \lambda_{AP}}{2\Gamma(N)}\sqrt{\frac{\pi^3 P_D\eta}{P_U}} \right)$,
which is equal to $1$ when $F$ is sufficiently large.
As a result, from (\ref{eq: trans_prob_def}) and (\ref{eq: appx_S_F_geq}), we
obtain
\begin{align}
\rho&\geq \lim_{n \to \infty} \frac{1}{n}
\sum_{F=1}^{n}\mathbb{P}\left(\sum_{i=1}^{F}Z_i \geq (F+1) P_U \right) \nonumber
\\
&=\lim_{n \to \infty} \frac{1}{n} \sum_{F=1}^{n}\textrm{erf}\left(\frac{\Gamma(N+2/\alpha) F \lambda_{AP}}{2\Gamma(N)}\sqrt{\frac{\pi^3 P_D\eta}{P_U}} \right) \nonumber \\
&=1.
\end{align}
Since $\rho\leq 1$, we have $\rho=1$. Proposition \ref{proposition:
rho_storage_Inf} thus follows.

\section{Proof of Proposition \ref{proposition: rho_bound_closedForm}}
\label{appendix: proof_rho_finite_B}
Since both $\textrm{erf}\left(\frac{\pi^2 \lambda_{AP} N}{4}
\sqrt{\frac{P_D\eta}{P_U}}\right)$ and $1-e^{-Q(C-P_U)}$ are lower bounds of
$\rho\leq 1$, it is easy to find that $\mathcal{L}\leq \rho \leq 1$
 holds. This proof mainly derives
the expression of $Q$ by solving
$\ln \mathbb{E}\left[e^{-Q(Z_F-P_U)} \right]=0$, or equivalently,
$\mathbb{E}\left(e^{-QZ_F} \right)=e^{-QP_U}$.
From the Laplace transform of $Z_F$ given in Proposition \ref{proposition: Z_F_distribution},
it is easy to find that $\mathbb{E}\left(e^{-QZ_F} \right)=\exp \left(-\pi \lambda_{AP} \frac{\Gamma(N+2/\alpha)}{\Gamma(N)} \Gamma(1-2/\alpha)(P_D \eta Q)^{2/\alpha} \right)$,
and thus $Q$ by letting $\mathbb{E}\left(e^{-QZ_F} \right)=e^{-QP_U}$.
 Proposition
\ref{proposition: rho_bound_closedForm} thus follows.


\begin{thebibliography}{11}


\bibitem{microwave} R. J. M. Vullers, R. V. Schaijk, I. Doms, C. V. Hoof, and R. Mertens,
``Micropower energy harvesting,'' {\it Elsevier Solid-State Circuits}, vol. 53,
no. 7, pp. 684-693, Jul. 2009.

\bibitem{Visser.IEEEProc.2013} H. J. Visser and R. J. Vullers,
``RF energy harvesting and transport for wireless sensor network applications:
Principles and requirements,'' {\it Proc. IEEE}, vol. 101, no. 6, pp. 1410–1423,
Apr. 2013.




\bibitem{EH_overview} A. M. Zungeru, L. M. Ang, S. Prabaharan, and K. P. Seng,
``Radio frequency energy harvesting and management for wireless sensor
networks,'' {\it Green Mobile Devices Netw.: Energy Opt. Scav. Tech.}, CRC
Press, pp. 341-368, 2012.

\bibitem{He.TMC.13} S.~He, J.~Chen, F.~Jiang, D.~Yau, G.~Xing, and Y.~Sun, ``Energy provisioning in wireless
rechargeable sensor networks,'' {\it IEEE Trans. Mob. Comput.}, vol.~12, no.~10, pp.~1931-1942,  Oct. 2013.

\bibitem{Ju.TWC} H.~Ju and R.~Zhang, ``Throughput maximization in
wireless powered communication networks,'' {\it IEEE Trans. Wireless Commun.},
vol.~13, no. 1, pp. 418-428, Jan. 2014.

%\bibitem{Xun} X. Zhou, R. Zhang, and C. K. Ho, ``Wireless information and power
%transfer: Architecture design and rate-energy tradeoff,'' {\it IEEE
%Trans. Commun.}, vol. 61, no. 11, pp. 4757-4767, Nov.  2013.

%\bibitem{Iannello.TCOM.2012} F.~Iannello, O.~Simeone, and U.~Spagnolini,
%``Medium access control protocols for
%wireless sensor networks with energy harvesting,'' {\it IEEE
%Trans. Commun.}, vol. 60, no. 5, pp. 1381-1389, May  2012.

\bibitem{Ozel.JSAC.2011}O.~Ozel, K.~Tutuncuoglu, J.~Yang, S.~ Ulukus,  and
A.~Yener,
``Transmission with energy harvesting nodes in fading wireless channels: Optimal
policies,'' \emph{IEEE J. Sel. Areas Commun.}, vol. 29, no. 8, pp.~1732-1743,
Sep. 2011.

\bibitem{Ho.TSP.2012} C.~K.~Ho and R.~Zhang, ``Optimal energy allocation for
wireless communications with energy harvesting constraints,'' \emph{IEEE Trans.
Sig. Process.}, vol.~60, no.~9, pp.~4808-4818, Sep.~2012.

\bibitem{Xu.JSAC.2014} J.~Xu and R.~Zhang, ``Throughput optimal policies for
energy harvesting wireless transmitters with non-ideal circuit power,''
\emph{IEEE J. Sel. Areas Commun.}, vol.~32, no.~2, pp. 322-332, Feb.~2014.


%\bibitem{Orhan.ISIT.2013}O.~Orhan, D.~G\"und\"uz, and E.~Erkip, ``Throughput maximization for an energy harvesting
%communication system with processing cost,''  in {\it Proc.  IEEE Int. Symp. Inf. Theory}, Jul. 2013.















\bibitem{Rui.TWC.13} R. Zhang and C. K. Ho, ``MIMO broadcasting for simultaneous wireless
information and power transfer,'' {\it IEEE Trans. Wireless Commun.}, vol.
12, no. 5, pp. 1989-2001, May 2013.

\bibitem{Liu.TWC.13} L. Liu, R. Zhang, and K. C. Chua, ``Wireless information transfer with
opportunistic energy harvesting,'' {\it IEEE Trans. Wireless Commun.}, vol.
12, no. 1, pp. 288-300, Jan. 2013.

\bibitem{Wen.datacenter1} Y.~C.~Jin, Y.~G.~Wen, and Q.~H.~Chen, ``Energy efficiency and server virtualization in
data centers: An empirical investigation,'' {\it IEEE INFOCOM Workshop on Communications and Control for Sustainable Energy Systems:
Green Networking and Smart Grids}, Mar. 2012.

\bibitem{Wen.datacenter2}Y.~C.~Jin, Y.~G.~Wen, Z.~Q.~Zhu, and Q.~H.~Chen, `` An empirical investigation of the impact
of server virtualization on energy efficiency for green data center,'' {\it The Computer Journal}, Oxford Journals,
vol.~51, no.~8, pp.~968-975, Aug.~2013.

\bibitem{Wen.EE} W.~W.~Zhang, Y.~G.~Wen, K.~Guan, D.~Kilper, H.~Y.~Luo, and D.~P.~Wu, ``Energy-efficient
mobile computing under stochastic wireless channel,'' {\it IEEE Trans. Wireless Commun.}, vol.
12, no. 9, pp. 4569-4581, Sep. 2013.





\bibitem{Stoyan.SG.95} D.~Stoyan, W.~S.~Kendall, and J.~Mecke, {\it Stochastic
geometry and its applications}, 2nd edition. John Wiley and Sons, 1995.


%\bibitem{Andrews.COM.09} S.~Weber, J.~G.~Andrews, and N.~Jindal, ``An overview
%of the transmission capacity of wireless networks,'' {\it IEEE Trans. Commun.},
%vol. 58, no. 12, pp. 3593-3604, Dec. 2010.

\bibitem{Haenggi.book} M.~Haenggi and R.~K.~Ganti, {\it Interference in large
 wireless networks.} NOW: Foundations and Trends in Networking, ~2008.

\bibitem{Huang.TWC.14} K.~Huang and V.~K.~N.~Lau, ``Enabling wireless power
transfer in cellular networks: architecture, modeling and deployment,''  {\it
IEEE Trans. Wireless Commun.}, vol. 13, no. 2, pp. 902-912, Feb. 2014.




\bibitem{Lee.TWC.2013} S.~Lee, R.~Zhang, and K.~Huang, ``Opportunistic wireless
energy harvesting in cognitive radio networks,'' {\it IEEE Trans. Wireless
Commun.}, vol. 12, no. 9, pp. 4788-4799, Sep. 2013.


\bibitem{Type_I}  Available [online] at http://witricity.com/applications/military.

\bibitem{Type_II} L. Xie, Y. Shi, Y. T. Hou, and H. D. Sherali, ``Making sensor networks immortal: an energy-renewable
approach with wireless power transfer,''  {\it IEEE/ACM Trans. Netw.},  vol. 20, no. 6, pp. 1748-1761, Dec. 2012.


\bibitem{Che.TWC.14} Y.~L.~Che, R.~Zhang, Y.~Gong and L.~Duan,
``On spatial capacity of wireless ad hoc networks with  threshold based
scheduling,'' {\it IEEE Trans. Wireless
Commun.},  vol. 13, no. 12, pp. 6915-6927, Oct. 2014.


\bibitem {Huang.IT.13} K.~Huang, ``Spatial throughput of mobile ad hoc networks
with energy harvesting,'' {\it IEEE Trans. Inf. Theory}, vol. 59, no. 11, pp.
7597-7612, Nov. 2013.

\bibitem{Dhillon.TWC.14} H.~S.~Dhillon, Y.~Li, P.~Nuggehalli, Z.~Pi, and J.~G.~Andrews, ``Fundamentals
of heterogeneous cellular networks with energy harvesting", {\it IEEE Trans. Wireless
Commun.}, vol.~13, no.~5, May 2014.


\bibitem{Sample.RFID.07} A.~P. Sample, D.~J.~Yeager, P.~S.~Powledge, and J.~R.~Smith, ``Design of a
passively-powered, programmable platform for UHF RFID systems,'' in
{\it Proc. IEEE Int. Conf. RFID,} Mar. 2007.





\bibitem{EH_system} A. Sample and J. R. Smith, ``Experimental results with two wireless
power transfer systems,'' in {\it IEEE Radio Wireless Symp.,} Jan.~2009.

\bibitem{Powercast} Product datasheet P2110-915 MHz RF
powerharvester$^\textrm{TM}$
receiver.
Available [online] at http://www.powercastco.com/PDF/P2110-datasheet.pdf.



\bibitem{Wen.scheduling}W.~W.~Zhang, Y.~G.~Wen, and D.~P.~Wu, ``Energy-efficient scheduling policy for
collaborative  execution in mobile cloud computing,'' in
\emph{Proc. IEEE Int. Conf. Computer Commun. (INFOCOM)}, Apr.~2013.

\bibitem{Andrews.COM.11} J.~G.~Andrews, F.~Baccelli, and R.~K.~Ganti, ``A
tractable approach to coverage and rate in cellular networks,'' {\it IEEE Trans.
Commun.},  vol. 59, no. 11, pp. 3122-3134, Nov. 2011.


\bibitem{Baccelli.NOW.I} F.~Baccelli and B.~Blaszczyszyn, {\it Stochastic Geometry and Wireless Networks, Volume I: Theory}.   NOW: Foundations and Trends in Networking, 2009.



\bibitem{Gong.TMC.14}Z.~Gong and M.~Haenggi, ``Interference and outage in mobile random networks: expectation, distribution, and correlation,'' {\it IEEE Trans. Mob. Comput.}, vol. 13, pp. 337-349, Feb. 2014.



 \bibitem{Weber.Alg} S.~P.~Weber and M.~Kam, ``Computational complexity of
outage probability simulations in mobile ad-hoc networks,'' in {\it Proc., Conf.
on Information Sciences and Systems}, Mar. 2005.


%\bibitem{Shadowing} H.~S.~Dhillon  and J.~G.~Andrews, ``Downlink rate distribution in heterogeneous cellular networks under generalized cell selection,'' {\it IEEE
%Wireless Commun. Letters}, vol.~ 3, no.~1, Feb.~2014.






\end{thebibliography}
\end{document}